\documentclass[a4paper,USenglish,cleveref]{lipics-v2019}

\hideLIPIcs
\nolinenumbers

\usepackage[ruled]{algorithm}
\usepackage{xspace}
\usepackage{ntabbing}

\usepackage{pgfplots}
\pgfplotsset{compat=newest}
\usetikzlibrary{patterns}
\usetikzlibrary{pgfplots.groupplots}

\newenvironment{customlegend}[1][]{%
  \begingroup\csname pgfplots@init@cleared@structures\endcsname\pgfplotsset{#1}%
}{\csname pgfplots@createlegend\endcsname\endgroup}%
\def\addlegendimage{\csname pgfplots@addlegendimage\endcsname}

\makeatletter
\newif\ifhidelinks@hidelinks
\newcommand{\hidelinks}{%
  \hidelinks@hidelinkstrue
  \let\hidelinks@ifHy@colorlinks@status\ifHy@colorlinks
  \let\hidelinks@ifHy@ocgcolorlinks@status\ifHy@ocgcolorlinks
  \let\hidelinks@ifHy@frenchlinks@status\ifHy@frenchlinks
  \let\hidelinks@Hy@colorlink\Hy@colorlink
  \let\hidelinks@Hy@endcolorlink\Hy@endcolorlink
  \let\hidelinks@@pdfborder\@pdfborder
  \let\hidelinks@@pdfborderstyle\@pdfborderstyle
  \hypersetup{hidelinks}%
}
\newcommand{\restorelinks}{%
  \ifhidelinks@hidelinks
    \hidelinks@hidelinksfalse
    \let\ifHy@colorlinks\hidelinks@ifHy@colorlinks@status
    \let\ifHy@ocgcolorlinks\hidelinks@ifHy@ocgcolorlinks@status
    \let\ifHy@frenchlinks\hidelinks@ifHy@frenchlinks@status
    \let\Hy@colorlink\hidelinks@Hy@colorlink
    \let\Hy@endcolorlink\hidelinks@Hy@endcolorlink
    \let\@pdfborder\hidelinks@@pdfborder
    \let\@pdfborderstyle\hidelinks@@pdfborderstyle
  \fi
}
\makeatother

\definecolor{plotcolor0}{RGB}{255,0,0}
\definecolor{plotcolor1}{RGB}{0,0,255}
\definecolor{plotcolor2}{RGB}{64,170,0}

\pgfplotscreateplotcyclelist{mylist}{%
  plotcolor0, every mark/.append style={solid,fill opacity=0.6,scale=1}, mark=diamond* \\%
  plotcolor1, every mark/.append style={solid,fill opacity=0.6,scale=0.8}, mark=square* \\%
  plotcolor2, every mark/.append style={solid,fill opacity=0.6,scale=1}, mark=* \\%
  plotcolor0, every mark/.append style={solid,fill opacity=0.2,scale=1}, mark=diamond*, dashed \\%
  plotcolor1, every mark/.append style={solid,fill opacity=0.2,scale=0.8}, mark=square*, dashed \\%
  plotcolor2, every mark/.append style={solid,fill opacity=0.2,scale=1}, mark=*, dashed \\%
  plotcolor0, every mark/.append style={solid,scale=1}, mark=diamond, densely dotted \\%
  plotcolor1, every mark/.append style={solid,scale=0.8}, mark=square, densely dotted \\%
  plotcolor2, every mark/.append style={solid,scale=1}, mark=o, densely dotted \\%
}

\newcommand{\ie}{\emph{i.e.,}\xspace}
\newcommand{\eg}{\emph{e.g.,}\xspace}
\newcommand{\cf}{\emph{cf.}\xspace}
\newcommand{\etal}{et~al.\@}

\newcommand{\Tstart}{\alpha}
\newcommand{\Tword}{\beta}
\newcommand{\Tsel}{T_{\textit{sel}}}
\newcommand{\kmin}{\underline{k}}
\newcommand{\kmax}{\overline{k}}
\newcommand{\rand}[1][]{\texttt{rand}(#1)}

\newcommand{\Oh}[1]{\mathcal{O}\!\left( #1\right)}
\newcommand{\Ohsmall}[1]{\mathcal{O}(#1)}
\newcommand{\Ohsmash}[1]{\smash{\Oh{#1}}}
\newcommand{\Om}[1]{\Omega\left(#1\right)}
\newcommand{\Th}[1]{\Theta\!\left( #1\right)}
\newcommand{\prob}[1]{{\mathbf{P}}\left[#1\right]}
\newcommand{\expect}[1]{{\mathbf{E}}\left[#1\right]}
\newcommand{\ceil}[1]{\left\lceil #1\right\rceil}
\newcommand{\floor}[1]{\left\lfloor #1\right\rfloor}
\newcommand{\card}[1]{\left\vert{#1}\right\vert}
\newcommand{\set}[1]{\left\{ #1\right\}}

\newcommand{\real}{\mathbb{R}}
\newcommand{\rplus}{\mathbb{R}_+}
\newcommand{\nat}{\mathbb{N}}

\newenvironment{code}{\noindent\normalsize%
\begin{ntabbing}%
\hspace{2em}\=\hspace{2em}\=\hspace{2em}\=\hspace{2em}\=\hspace{2em}\=%
\hspace{2em}\=\hspace{2em}\=\hspace{2em}\=\hspace{2em}\=\hspace{2em}\=%
\kill}{\end{ntabbing}}%
\newcommand{\BR}[1]{\texttt{[}#1\texttt{]}}
\newcommand{\DeclareInit}[3]{#1 $\Def$ #3\ \rm : #2}
\newcommand{\Declare}[2]{#1\mbox{ \rm : }#2}
\newcommand{\Def}{:=}
\newcommand{\Do}       {\textbf{do\ }}
\newcommand{\Else}     {\textbf{else\ }}
\newcommand{\Foreach}      {\textbf{foreach\ }}
\newcommand{\Function} {\textbf{def\ }}
\newcommand{\Funct}[3]{\Function #1\Declare{{\rm (}{#2\rm )}}{#3}}
\newcommand{\If}       {\textbf{if\ }}
\newcommand{\Return}   {\textbf{return\ }}
\newcommand{\Then}     {\textbf{then\ }}
\newcommand{\While}    {\textbf{while\ }}
\newcommand{\Comment}[1]   {\`\textit{--- #1}}

\hypersetup{
  colorlinks=true,
  pdftitle={Communication-Efficient (Weighted) Reservoir Sampling from Fully Distributed Data Streams},
  pdfauthor={Lorenz Hübschle-Schneider and Peter Sanders},
  pdfkeywords={Sampling, Weighted sampling, Reservoir sampling, Mini-batch, Data stream, Distributed streaming, Communication efficiency}
}

\title{Communication-Efficient (Weighted) Reservoir Sampling from Fully Distributed Data Streams}

\author{Lorenz Hübschle-Schneider}{Karlsruhe Institute of Technology, Germany}{huebschle@kit.edu}{}{}
\author{Peter Sanders}{Karlsruhe Institute of Technology, Germany}{sanders@kit.edu}{}{}
\authorrunning{L. Hübschle-Schneider and P. Sanders}
\titlerunning{Communication-Efficient (Weighted) Reservoir Sampling from Fully Distributed \ldots}

\Copyright{Lorenz Hübschle-Schneider and Peter Sanders}



\supplement{The code and scripts used for our experiments are available under
  the GPLv3 at \url{https://github.com/lorenzhs/reservoir}.}

\begin{document}

\maketitle

\begin{abstract}
  We consider communication-efficient weighted and unweighted (uniform) random sampling from distributed data
  streams presented as a sequence of mini-batches of items.  This is a natural
  model for distributed streaming computation, and our goal is to showcase its
  usefulness.  We present and analyze fully distributed, communication-efficient
  algorithms for both versions of the problem.  An experimental evaluation of
  weighted reservoir sampling on up to 256 nodes (5120 processors) shows good
  speedups, while theoretical analysis promises further scaling to much larger
  machines.
\end{abstract}

\section{Introduction}

Random sampling is a fundamental tool used by many algorithms (see, \eg Cormen
\etal~\cite{Cormen}).  In its \emph{uniform} or \emph{unweighted} form, each
item has the same probability of being picked, while \emph{weighted} sampling
associates a weight with each item, and items are picked with probability
proportional to their share of the total weight.  Here, we consider a setting
where the input is not known in advance, but rather arrives over time and
distributed over many machines (called \emph{nodes}).  We ask that the nodes
collaboratively maintain a sample \emph{without replacement} of all items seen
so far at any node.  This is motivated by the need to work with ever-larger data
sets that are too large to store, even if every machine were to store only a
part of it.  Applications include databases~\cite{olken1995survey}, search
engines, data mining, network monitoring, and large-scale web applications in
general~\cite{msft07distres}.  More fundamentally, it is an ingredient in many
other randomized algorithms, such as maintaining the set of \emph{heavy hitters}
or \emph{count tracking} (\eg~\cite{JSTW19wres}).

In designing such algorithms, we strive for \emph{communication efficiency},
that is, to minimize communication between the nodes.  This is motivated by the
design of real-world data centers and supercomputers, where communication is an
expensive resource and highly underprovisioned.  For a more detailed motivation
of communication efficiency, refer to Sanders \etal~\cite{SSM13,HubSan16topk}. %

To achieve communication efficiency, an algorithm has to be fully distributed.
Therefore, models of computation that assign some node a special
\emph{coordinator} role may introduce unnecessary bottlenecks.  Most
prominently, the \emph{distributed streaming model}, also known as
\emph{continuous monitoring model (CMM)} \cite{cormode2013monitoring} relies on
such a centralized coordinator.  The goal is for the coordinator to continuously
maintain the result of the computation with as few (fixed-size) messages
exchanged between it and the other nodes (called \emph{sites}).  The sites
cannot communicate with one another directly, only with (or via) the
coordinator.  In addition, the \emph{CMM} assumes per-item synchronization of
the nodes and synchronous operation of the network, which would be prohibitively
expensive in the real world, if not outright impossible.  As a result, the
scalability of algorithms in the \emph{CMM} is inherently limited by the load on
the coordinator.

\subsection{Mini-Batch Model}\label{model}

Instead of processing each item individually, we process items as a series of
\emph{mini-batches} which arrive on small time intervals and in a distributed
fashion: some of their items arrive at each processing element (PE).  They
might, for example, be delivered over the network or be read in blocks from a
file system.  Apache Spark Streaming defines them as the set of all items that
arrived within a certain time window since the previous batch finished~\cite{sparkstream2013}.  Because memory is limited, only the items of the
current mini-batch are available in memory at each point in time.
This is a natural generalization of multiple other models of streaming
algorithms.  On a sequential machine with batches of size 1, we obtain the
sequential streaming model (see, \eg \cite{TT19parstream}).  In a distributed
model with $p$ sites (nodes) which exchange fixed-size messages with a
coordinator, batches of a single item per site yield the continuous monitoring
model described above.  In this paper, we use a distributed message-passing
model (\cf \cref{prel}).  We use the terms \emph{batch} and \emph{mini-batch}
interchangeably.

\subsection{Problem Definition}\label{intro:def}

Let $k$ be the desired sample size, and let the input data set consist of $n$
items, which we shall refer to by their indices from $1..n$\footnote{$a..b$ is
  shorthand for $\set{a,\ldots,b}$ throughout this paper.} purely for notational
convenience.  For weighted sampling, additionally associate with each item $i$ a
\emph{weight} $w_i \in \rplus$, and let $W \Def \sum_{i=1}^{n}{w_i}$ denote the
total weight.  The items are processed in batches of variable size.  After
processing a batch, update the sample to be a uniform (or weighted) random
sample without replacement of size $\min(k, n')$ of all $n'$ items seen so far,
up to and including the current batch.

A \emph{uniform} or \emph{weighted random sample without replacement} of size
$k$ consists of $k$ pairwise different items $s_1 \neq \ldots \neq s_k$ such
that for any sample $s_j$ and item $i \notin \{s_\ell\mid \ell<j\}$, we have
$\prob{s_j = i}=1/(n-j+1)$ in the uniform case or
$\prob{s_j=i}=w_i/(W-\sum_{\ell<j}{w_{s_\ell}})$ in the weighted setting.

Note that a second definition of the weighted problem also exists, where the
probability of each item to be \emph{included in} the sample is proportional to
its relative weight.  In that definition, \emph{infeasible items} with relative
weight exceeding $1/k$ have inclusion probability greater than $1$ and require
special treatment (see, \eg Efraimidis~\cite[Example 2]{efraimidis2015stream}).
 We henceforth refer to that definition as \emph{weighted sampling with
   probabilities}.
The definition we use does not have the problem of infeasible items.

\section{Related Work}\label{rel}

For an overview of the broader literature concerning random sampling, we refer
the reader to Sanders \etal~\cite{SandersLHSD18} for the uniform case and
our recent paper on weighted sampling~\cite{HubSan19wrs} for the weighted setting.  Here, we
limit ourselves to the literature on reservoir sampling.

\subparagraph*{Uniform Reservoir Sampling.}\label{rel:unw}
Sampling from a \emph{stream of data} has been studied
since at least the early 1960s \cite{fan62res} and several asymptotically
optimal algorithms are known, see \eg Vitter~\cite{Vit85} and
Li~\cite{li1994reservoir}.  The key insight of these algorithms is that it is
possible to determine in constant time how many items to skip before a new item
enters the reservoir by computing a random deviate from a suitably parameterized
geometric distribution \cite{devroye1986,li1994reservoir}.

Sampling from the union
of multiple data streams was only studied much more recently
\cite{chung2016simple,cormode2012continuous,woodruff2011optimal,cormode2010optimal}.
However, these publications use the continuous monitoring model (see
Introduction), and are thus inherently limited in their scalability.
Recently, Tangwongsan and Tirthapura~\cite{TT19parstream} presented a shared-memory
parallel uniform reservoir sampling algorithm in a mini-batch model.
Birler~\cite{birler19res} proposes an approach using
synchronized access to a shared source of skip values.

\subparagraph*{Weighted Reservoir Sampling.}\label{rel:weighted}

Sampling from a stream of \emph{weighted} items has received significantly less
attention in the literature.  Chao \cite{chao1982reservoir} presents a simple
and elegant algorithm for weighted reservoir sampling \emph{with probabilities}
(see \cref{intro:def}).  Efraimidis and Spirakis give an algorithm based on
associating a suitably computed key with each item, so that the $k$ items with
the largest keys form the desired sample, and also show how to compute skip
distances (which they call \emph{exponential jumps}) in constant time
\cite{efraimidis2006reservoir,efraimidis2015stream}.  Braverman
\etal~\cite{braverman2015float} present an approach they call \emph{Cascade
  Sampling} that is not affected by the numerical inaccuracies of floating-point
representation in computers by giving an exact reduction to sampling \emph{with}
replacement.

The first -- and, to the best of our knowledge, only -- distributed streaming
algorithm for weighted reservoir sampling was published only recently by Jayaram
\etal~\cite{JSTW19wres}.  Their algorithm is given in the continuous monitoring
model, resulting in complex algorithmic challenges and requiring a \emph{level
  set} construction.

\section{Preliminaries}\label{prel}

Unless explicitly specified, we shall make no assumptions about the distribution
of mini-batch sizes across PEs or over time, nor about the distribution of
items.  In algorithm analysis, we denote by $b$ the maximum number of items in
the current mini-batch at any PE, and by $B$ the sum of all PEs' current
mini-batch sizes.  Thus, an algorithm expressed in the mini-batch model can
handle arbitrarily imbalanced inputs without any impact on correctness; however,
load balance may suffer if the number of items per PE differs widely.

\subparagraph*{Machine Model.} Consider $p$ processing elements (PEs) numbered
$1..p$ connected by a network so that each PE can send and receive at most one
message simultaneously to any other PE (full-duplex, single-ported
communication), which matches the machine model of Dietzfelbinger
\etal~\cite{MehSanPar}.  Sending a message of length $\ell$ takes time
$\Tstart + \Tword \ell$, where $\Tstart$ is the time to initiate the transfer
and $\Tword$ the time to send a single machine word once the connection has been
established.  By treating $\Tstart$ and $\Tword$ as variables in asymptotic
analysis, we can concisely combine \emph{internal work} $x$, \emph{communication
  volume} $y$ and \emph{latency} $z$ into a single asymptotic term for the time
of the critical path: $\Oh{x + \Tword y + \Tstart z}$.

\subparagraph*{Collective Communication.}
We use several fundamental communication operations that involve all PEs.  A
\emph{broadcast} distributes a message from one PE to all others.  A
\emph{reduction} applies an associative operation to the values of one or
several variables at the PEs, the result of which is available at a single
designated PE.  In an \emph{all-reduction}, the result of the reduction is made
available at every PE. %
All of these operations can be performed in
time $\Oh{\Tword \ell + \Tstart \log p}$ for a message of~$\ell$ machine words
\cite{BalEtAl95,SST09}. A \emph{gather} operation communicates one or more
values from each PE to a designated root PE.  This can be done in
$\Oh{\Tword p \ell + \Tstart \log p}$ time \cite[Chapter 13.5]{MehSanPar}, where
$\ell$ is the number of machine words sent by each PE.

\subsection{Sampling by Sorting Random Variates}\label{prel:exp}

It is well known that an unweighted sample without replacement of size $k$ out
of $n$ items $1..n$ can be obtained by associating with each item a uniform
variate, and selecting the $k$ items with the smallest associated variates (see,
\eg \cite{fan62res} for an early algorithm based on this idea).  This method can
be generalized to computing a \emph{weighted} sample without replacement by
raising uniform variates to the power of the inverse of the items' weights --
$v_i \Def \rand^{1/w_i}$, where $\rand$ generates a uniformly random deviate
from the interval $(0,1]$ -- and selecting the $k$ items with the \emph{largest}
associated values
\cite{efraimidis1999par,efraimidis2006reservoir,efraimidis2015stream}.
Equivalently, one can generate \emph{exponentially distributed random variates}
$v_i \Def -\ln(\rand)/w_i$ and select the $k$ items with the \emph{smallest}
associated $v_i$ \cite{arratia2002expclock,HubSan19wrs} (``\emph{exponential
  clocks method}''), which is numerically more stable.

\subsection{Search Trees}\label{prel:st}

A B+ tree is a search tree where the inner nodes only store keys and the leaf
nodes store the items (\ie key-value pairs).  The nodes have an arbitrary but
fixed maximum degree $d$; each node except the root is at least half full at any
point in time, \ie inner nodes have at least $\ceil{d/2}$ children and leaf
nodes store at least $\ceil{d/2}$ items.  By linking the leaf nodes it is always
possible to find the next-larger or next-smaller item in constant time.  As with
B-trees, it is possible to implement join and split operations in $\Oh{\log n}$
time for trees of $n$ items (\eg \cite[Chapter 7.3.2]{MehSanPar}).  By
additionally keeping track of subtree sizes, rank and select queries can be
answered in $\Oh{\log n}$ time as well (\eg \cite[Chapter 7.5.2]{MehSanPar}).  A
rank query asks how many items in the tree have smaller keys, and a select$(k)$
query asks for the item with the $k$-th smallest key.

\subsection{Selection from Sorted Sequences}\label{prel:sel}
\newcommand{\selsize}{g}
\newcommand{\selrank}{r}

Let each PE hold a sorted sequence of items, and let $\selsize$ be
an upper bound on the number of local items at any of the $p$ PEs.  We wish to
select the item with global rank $\selrank$ (\ie the item with the globally $\selrank$-th
smallest key).  Which selection algorithm to use depends on the data and our
requirements.  In the remainder of the paper, we will use $\Tsel$ as a
placeholder for the running time of an appropriate selection operation as
defined below for various cases.

\subsubsection{Randomly Distributed Items}\label{prel:sel:rand}
If the input is randomly distributed (\eg all input items are independently
drawn from a common distribution), we can use the algorithm of Sanders
\cite[Lemma 7]{San98a}.  It is based on sorting a sample of $\sqrt{p}$ items and
choosing two pivots such that the key of the item with rank $\selrank$ is one of
a small number of items between the two pivots with sufficiently high
probability.  Taking the sample of $\sqrt{p}$ items can be done in a
communication-efficient manner using Algorithm~P of Sanders \etal~\cite[Sections
3.2 and 4.6]{SandersLHSD18}, which in this case takes time $\Oh{\log p}$ with
high probability.  Each PE holds at most a constant number of samples with high
probability.  When using B+ trees as sorted sequence data structure, selecting a
sampled item from a local B+ tree takes $\Oh{\log(\selsize)}$ time (see
\cref{prel:st}).  Sorting the sample with the fast inefficient sorting algorithm
of Axtmann \etal\ takes time $\Oh{\Tstart \log p}$~\cite{ABSS15}.  Thus, the
entire selection process takes time $\Ohsmall{\log \selsize + \Tstart \log p}$
with high probability \cite{San98a}.

\subsubsection{Approximate Selection}\label{prel:sel:approx}

If the output rank $\selrank$ is allowed to vary in some range
$\underline{\selrank}..\overline{\selrank}$, efficient selection is possible
even if the input is not randomly distributed.  If this variation is large
enough, \ie
$\overline{\selrank}-\underline{\selrank} \in \Om{\overline{\selrank}}$, it is
possible to do selection with expected constant recursion depth.  In that case,
selection is possible in expected time
$\Oh{\log\selsize + \Tstart \log p}$~\cite[Theorem 2, Section
IV-C]{HubSan16topk}.

This algorithm can also be used with more than one pivot element in each level
of the recursion, reducing the gap that is required between
$\underline{\selrank}$ and $\overline{\selrank}$ to achieve constant expected
recursion depth.  When using $d \in \nat$ pivots,
$\overline{\selrank}-\underline{\selrank} \in \Om{\overline{\selrank}/d}$
suffices.  In that case, the expected running time is
$\Oh{d \log \selsize + \Tword d + \Tstart \log p}$ \cite[Lemma 3, Section
IV-C]{HubSan16topk}.

\subsubsection{General Case}\label{prel:sel:general}

If neither of the above special cases is applicable, we can fall back to an
algorithm with expected running time
$\Oh{\log (\selrank p) \cdot \log \selsize + \Tstart \log^2(\selrank
  p)}$~\cite[Section IV-B]{HubSan16topk}.  However, we can improve upon that
algorithm in practice by using the approximate selection algorithm of the
previous paragraph with $\underline{\selrank} = \overline{\selrank} = \selrank$.
In this special case, the algorithm uses the globally smallest item in a
Bernoulli sample with success probability $1/\selrank$ as the pivot element (or
the largest item with success probability $1/(G-\selrank+1)$ if $\selrank$ is
large with regard to the global input size $G$).  This item has expected rank
$\selrank$.  Then it computes the number $\selrank'$ of items less than or equal
to the pivot using an all-reduction. If $\selrank' = \selrank$ return the
pivot's item, if $\selrank' < \selrank$, recursively select the
$\selrank - \selrank'$-th smallest from the items larger than the pivot, else
recurse on the items whose keys are at most as large the pivot.  A variant of
this uses multiple pivot elements as described in \cref{prel:sel:approx} above,
similarly reducing the expected recursion depth.

If $\selrank p$ is very large, and $\Oh{\log^2(\selrank p)}$ latency is
undesirable, we can trade it for increased communication volume by combining the
algorithm of \cref{prel:sel:approx} with random redistribution of the result and
\cref{prel:sel:rand}.  We begin by using the approximate selection algorithm
with $d=1$, $\underline{\selrank}=\selrank$, and
$\overline{\selrank}=2\selrank$.  This takes expected time
$\Oh{\log \selsize + \Tstart \log p}$. Let $v$ be the key of the selected
item. In a second step, all $\Th{\selrank}$ items with keys smaller than $v$ are
randomly redistributed, \eg with a hypercube all-to-all exchange (see, \eg
\cite[Chapter 13.6.2]{MehSanPar}).  This is possible in time
$\Oh{(\Tword \selrank + \Tstart) \log p}$ with high probability (whp) because no
PE ends up with more than $\Th{\selrank}$ items whp by a standard
balls-into-bins analysis.  Now we can apply the algorithm of
\cref{prel:sel:rand} to the shuffled items, which takes time
$\Oh{\log \selsize + \Tstart \log p}$ with high probability.  Overall, the
expected running time is then
$\Oh{\log \selsize + (\Tword \selrank + \Tstart) \log p}$.

Lastly, it is also possible to use a selection algorithm for unsorted inputs.
In a previous paper, we presented an algorithm with
$\Oh{\smash{\selsize} + \Tword \min(\sqrt{p}(1 + \log_p \selsize),\,
  \smash{\selsize}) + \Tstart \log \selsize}$ running time with high probability
\cite[Section IV-A]{HubSan16topk}.

\section{Reservoir Sampling}\label{res}

The basic idea of our algorithm is to maintain a distributed reservoir using a
communication-efficient bulk priority queue~\cite{HubSan16topk}.  Each PE holds
those items of the sample that were seen in its input, and no PE gets a special
role (\eg coordinator).

First, we adapt the skip value computation of Efraimidis and
Spirakis~\cite{efraimidis2006reservoir} in \cref{res:skip}, yielding a
sequential weighted reservoir sampling algorithm, before introducing our main
algorithm in \cref{res:w}.  \Cref{res:u} gives a straight-forward adaptation to
uniform reservoir sampling, and \cref{res:var} describes an optimization if the
sample size is allowed to vary in a fixed interval.  Lastly, \cref{res:gather}
outlines a %
more centralized approach as a comparison point.

\subsection{Skip Values (``Exponential Jumps'')}\label{res:skip}

Efraimidis and Spirakis~\cite[Section 4]{efraimidis2006reservoir} show how to
compute the amount of weight that does not yield a sample, \ie how much weight
to skip until an item is sampled.  They refer to this
technique as \emph{exponential jumps}.  Here, we show how to adapt their method
to the exponentially distributed variates described in \cref{prel:exp}.  This allows for faster
and numerically more efficient generation in practice.  The difference between
the variates associated with items in the method of Efraimidis and
Spirakis~\cite{efraimidis2006reservoir} and here is a simple $x \mapsto\nobreak -\ln(x)$
mapping.  Because of the sign inversion, the reservoir $R$ now contains the items
with the \emph{smallest} associated keys.  Let $v_i$ denote the key of item
$i$, that is, the exponentially distributed variate associated with it, and
define $T \Def \max_{i \in R}v_i$ as the threshold value, \ie the largest key of
any item in the reservoir.  Then the skip value $X$ describing the total amount
of weight to be skipped before an item next enters the reservoir can be computed
as $X \Def -\ln(\rand)/T$.  This is an exponential random deviate with rate parameter~$T$.
The key associated with the item $j$ that is to be inserted into the reservoir
is then $v_j \Def -\ln(\rand[e^{-Tw_j},1])/w_j$, where
$\rand[a,b]\Def a+\rand(b-a)$ generates a uniform random
variate from the interval $(a,b]$.  The range of this variate has been chosen so
that $v_j$ is less than $T$ (as at this stage, it has already
been determined that item $j$ must be part of the reservoir, we need to compute
a suitable variate from the distribution associated with its weight).  Item $j$ then
replaces the item with the largest key in the reservoir, and the threshold $T$
is updated to the now-largest key of any item in the reservoir.

\subsection{Fully Distributed Weighted Reservoir Sampling}\label{res:w}

\begin{figure}
  \centering
  \includegraphics{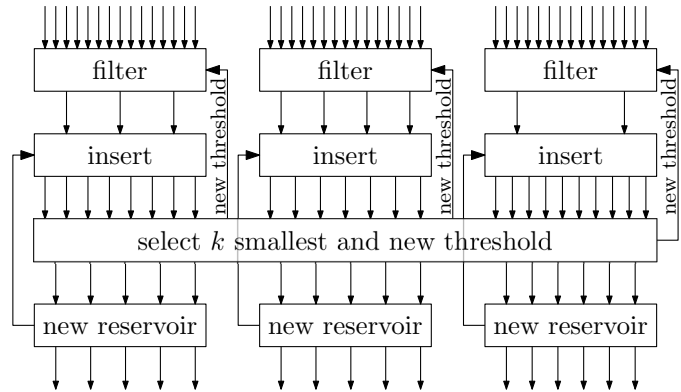}
  \caption{Schematic view of data flow in our algorithm of \cref{res:w}.}\label{fig:schema}
\end{figure}

The basic idea of our distributed algorithm is to combine this sequential algorithm with a
bulk distributed priority queue to maintain the set of the $k$ items with the
globally smallest keys, \ie the sample.  Each PE is solely responsible for the
items in its input, and no PE has a special role (such as a coordinator node).
During each batch, a PE inserts into its \emph{local} reservoir all items whose
key is smaller than some threshold.  When a batch finishes, the PEs perform a
distributed selection for the $k$-th smallest key, which becomes the new
threshold, and discard all items with larger keys.  The remaining elements form
the new sample.

We now flesh out the remaining parts of our algorithm.  Firstly, the reservoir
is maintained in a distributed fashion.  Each PE's \emph{local reservoir} is a
B+ tree that is augmented to also support the \emph{split}, \emph{rank}, and
\emph{select} operations in logarithmic time (see \cref{prel:st}).  The
\emph{split} operation is used to quickly discard the items that are no longer
part of the sample at the end of a batch, while \emph{rank} and \emph{select}
are required for the threshold selection process.

Secondly, the algorithm \emph{globally} maintains the threshold $T$, which is
the same at all PEs and does not change during a mini-batch.  The PEs process
their items using the skip distance method described in \cref{res:skip},
inserting the new candidate sample items into their local reservoirs.

Lastly, once all items of the mini-batch are processed, the PEs jointly select
the globally $k$-th smallest key (see \cref{prel:sel}) in the union of all local
reservoirs.  This key becomes the insertion threshold for the next mini-batch.
Each PE then discards all items with larger keys using a \emph{split} operation
on its local reservoir.  The remaining $k$ items in the union of all local
reservoirs are a weighted sample without replacement of size $k$ of all items
seen so far.  Algorithm \ref{alg:w} gives pseudocode and \cref{fig:schema} shows
a schematic view of the algorithm.

\begin{algorithm}[bt]
\vspace*{-0.5em}
\begin{code}
  \Funct{processBatch}{%
    \Declare{$A$}{Item\BR{}},
    \Declare{$T$}{$\real$},
    \Declare{$R$}{Reservoir}}{$\real \times \text{Reservoir}$}\+\label{}\\
  Item: $\rplus \times \nat$ with weight $w\in\rplus$, index $i\in\nat$\label{}\\
  Reservoir: B+ tree mapping keys from $\real$ to item IDs\label{}\\
  \If{$T < 0$} \Then \Comment{fewer than $k$ items seen globally before this batch}\+\label{}\\
    \Foreach{$(w,i) \in A$} \Do\+\\
      $R$.insert($-\ln(\rand)/w$, $i$)\-\-\Comment{exponentially distributed keys}\label{}\\
  \Else\+\label{}\\
    \DeclareInit{j}{$\nat$}{$0$};\Comment{1-based index of next item, initially invalid}\label{}\\
    \While $j\leq\card{A}$ \Do\+\label{}\\
      \DeclareInit{$X$}{$\real$}{$-\ln(\rand)/T$}\Comment{weight to be skipped}\label{}\\
      \While $X>0$ \Do\Comment{skip $X$ amount of weight in total}\+\label{}\\
        $j \Def j+1$\label{}\\
        \If $j > \card{A}$ \Then break from both loops\label{c:break}\\
        $X \Def X-A[j].w$\-\label{}\\
      \DeclareInit{$x$}{$\real$}{$\exp(-T \cdot A[j].w)$}\label{}\\
      \DeclareInit{$v$}{$\real$}{$-\ln(\rand[x,1])/A[j].w$}\Comment{new key}\label{}\\
      $R$.insert($v$, $A[j].i$)\-\-\label{}\\

  \DeclareInit{$(T,i)$}{$\real \times \nat$}{select($R$, $k$)}\Comment{select $k$ globally smallest and new threshold}\label{}\\
  $(R, \_) \Def R$.splitAt($i$) \Comment{discard local items with larger keys}\label{}\\
  \Return $(T, R)$ \Comment{return new threshold and reservoir}\label{}
\end{code}
 \vspace*{-0.5em}
\caption{Pseudocode for weighted reservoir sampling in a Single-Program
  Multiple-Data style.  Inputs: local portion of the mini-batch $A$, old
  threshold $T$ (initially $-\infty$), and local reservoir $R$ (initially
  empty).}\label{alg:w}
\end{algorithm}

\begin{theorem}\label{thm:res}
  For weighted reservoir sampling with sample size $k$, processing a mini-batch
  of up to $b$ items is possible in time $\Oh{b + (b^*+1)\log(b^*+k) + \Tsel}$,
  where $b^*\leq b$ is the maximum number of items from the mini-batch on any PE
  that is below the insertion threshold, and $\Tsel$ is the time for selection
  from sorted sequences of size at most $b^*+k$ per PE (see \cref{prel:sel}).
\end{theorem}
\begin{proof}
    By definition of $b^*$, the local insertions require time
  $\Oh{b^*\log(b^*+k)}$ in total because each local reservoir has size at most
  $k$ at the start of the batch.  Since we have to process each item's weight
  even when using the above skip value technique, $\Oh{b}$ time is required to
  identify the items to be inserted into the reservoir.  The selection operation
  takes time $\Tsel$ which varies depending on the specifics of the input.  The
  number of candidate items per PE for the selection is clearly bounded by the
  local reservoir size of at most $k+b^*$.  The split operation to discard the
  items with keys exceeding the new threshold takes time logarithmic therein,
  \ie $\Oh{\log(k+b^*)}$ time.

  Now consider the implications of splitting the single stream of the sequential
  case into multiple independently-handled streams without carrying over any
  remaining skip weight $X$ at the end of the stream (line \ref{c:break} of
  \cref{alg:w}).  This maintains correctness because each unit of weight has
  the same probability of spawning a sample \emph{by design}, regardless of when
  the procedure was started.

  Keeping the threshold fixed for the duration of each mini-batch also maintains
  correctness because the set of all items above \emph{any} threshold always
  forms a weighted random sample (whose size is not known a priori).  Here, the
  threshold is chosen so that the sample size is guaranteed to be at least $k$,
  and the selection process determines a new threshold to restore a fixed sample
  size of $k$ items.
\end{proof}

Observe that for the case of randomly assigned items, $\Tsel$ is $\Oh{\log(b^*+k) + \Tstart \log p}$
(see \cref{prel:sel:rand}), which simplifies to an $\Tstart \log p$
contribution to the bound of \cref{thm:res}.

Furthermore, when using the combined selection method of
\cref{prel:sel:general}, only the new candidates (up to $b^*$ per PE) have to be
redistributed.  The first, approximate, selection can also be skipped in most
cases, because the total number of new candidates is very unlikely to exceed
$k$, especially in later batches.  Thus, we can simply count the total number of
new candidates -- call it $c\leq b^*p$ -- with an all-reduction in time
$\Oh{\Tstart \log p}$.  If $c\leq k$, we can directly proceed with
redistribution \emph{of the $c$ new candidates}.  The items that were already in
the reservoir before this round do not need to be redistributed because they
were already randomly shuffled before.  This reshuffling of relatively few items
can be done using a hypercube all-to-all exchange (\eg \cite[Chapter
13.6.2]{MehSanPar}, see also \cref{prel:sel:general}). If $b^*$ is small (and
therefore, so is $c$), a \emph{sparse asynchronous all-to-all}, using
non-blocking sends and a non-blocking barrier, could prove faster in practice.
Finally, the $k$ smallest keys from the union of the old sample and the new
candidates are selected using the algorithm of \cref{prel:sel:rand}.  In total,
this method takes time
$\Tsel \in \Oh{\log(b^*+k) + \Tword \min(b^*,k)\log p + \Tstart \log p}$ for
selection with hypercube all-to-all, or
$\Tsel \in \Oh{\log(b^*+k) + \Tstart (b^* + \log p)}$ with a sparse asynchronous
all-to-all.

The question now is how many items we (unnecessarily) insert into the local
reservoirs by leaving the threshold unchanged during a mini-batch.  We first
analyze the number of insertions into \emph{each} PE's local reservoir in
Lemma~\ref{lem:load}, before considering the expected maximum number of insertions
into \emph{any} PE's local reservoir in \cref{thm:load}.

\begin{lemma}\label{lem:load}
  If the item weights are %
  independently drawn from a common
  continuous distribution %
  and all batches have the same number of items on every PE,
  then our algorithm inserts
  $\Ohsmall{\frac kp(1+\log\frac{n}{k})}$ items into each local reservoir in
  expectation.
\end{lemma}
\begin{proof}
  Efraimidis and Spirakis~\cite[Proposition 7]{efraimidis2006reservoir} show that if
  the $w_i$ are independent random variates from a common continuous
  distribution, their sequential reservoir sampling algorithm inserts
  $\Oh{k \log(n/k)}$ items into the reservoir in expectation.  We adapt this to
  mini-batches of $b$ items per PE. Let $X_i$ denote the number of insertions on
  a PE for batch~$i$.  We obtain a binomially distributed random variable with
  expectation
  \[
    \expect{X_i} = \sum_{j=1}^{b}{\prob{\text{item $j$
          is inserted}}} = b\cdot\frac{k}{n_\mathrm{pre}} \leq b\cdot\frac{k}{ipb}=\frac{k}{ip},
  \]
  where $n_\mathrm{pre}$ is the number of items seen globally before the batch
  began.  For the initial $i_0=\frac{k}{bp}$ iterations, this probability
  exceeds one, which we for with $b$ insertions per PE, \ie
  $b\cdot\frac{k}{bp}=k/p$ insertions overall.  For mini-batches
  $i_0 \leq i < \frac{n}{pb}$ we obtain
  \begin{align*}
    \expect{\sum X_i} & \leq \! \sum_{\frac{k}{bp}\leq i\leq \frac{n}{bp}}{\frac{k}{ip}}
                        = \frac{k}{p} \! \sum_{\frac{k}{bp}\leq i\leq \frac{n}{bp}}{\frac1i}
                        = \frac{k}{p}\left(H_{\frac{n}{bp}}-H_{\frac{k}{bp}}\right)\\
                      & \leq \frac{k}{p}\left(1+\ln\frac{n}{bp}-\ln\frac{k}{bp}\right)
                        = \frac kp\left(1+\ln\frac{n}{k}\right),
  \end{align*}
  where $H_n$ is the $n$-th harmonic number.
\end{proof}

\begin{theorem}\label{thm:load}
  If the item weights are %
  independently drawn from a common
  continuous distribution %
  and all batches have the same number of items on every PE,
  then our algorithm inserts no more than
  $\Oh{\frac kp \log\frac{n}{k} + \log p}$ items into \emph{any} local
  reservoir in expectation.
\end{theorem}
\begin{proof}
To obtain the maximum load over all PEs, we apply a Normal approximation to the
bound on the $X_i$ from the proof of Lemma~\ref{lem:load}, obtaining
$Y_i \sim \mathcal{N}\left(\frac{k}{ip},
  \frac{k}{ip}\left(1-\frac{k}{ipb}\right)\right)$.  Summing these over the
mini-batches as above, we again obtain a Normal distribution whose mean and
variance are the sum of its summands' means and variances.  We then apply a
bound on the maximum of i.i.d Normal random variables obtained using the
Cramér-Chernoff method \cite[Chapter 2.5]{BLM13concineq},
$\expect{\max_{j=1..p}Z_j} \leq \mu + \sigma\sqrt{2\ln p}$ for
$Z_j \sim \mathcal{N}(\mu, \sigma^2)$.  Using the mean of
the $Y_i$ as an upper bound to their variance, we obtain
$\mu + \sqrt{2\mu\ln p}$ as an upper bound to the maximum per-PE load for
$\mu = \frac kp (1+\ln\frac nk)$.  Thus, the expected \emph{bottleneck} number
of insertions into any local reservoir is
$\Ohsmall{\frac kp \log \frac nk + \log p}$.
\end{proof}

\subsection{Uniform Reservoir Sampling}\label{res:u}

The above algorithm can be easily adapted to uniform items by using the
well-known skip distances for uniform reservoir sampling \cite[p. 640,
\emph{``Reservoir sampling with geometric jumps''}]{devroye1986}.
Here, we adapt Devroye's algorithm to our notation and model.  Initially, when no threshold
is known, the keys of the items are simply uniform random deviates between 0 and
1, generated by $\rand$.  The number of items to be skipped then follows a
geometric distribution with success probability $T$ and can be computed as
$X \Def \floor{\ln(\rand)/\ln(1-T)}$ for a given threshold $T$.  The key of the
$X+1$-th item, which is inserted into the local reservoir, is then simply
$v \Def \rand \cdot T$.  The remainder of the algorithm works
analogously to the weighted case.  Note that skipping $X$ items is a
constant-time operation, whereas skipping $X'$ amount of weight in the weighted
case requires examining every item that is skipped.  As a result, the asymptotic
local processing time for a batch of uniform items is the number of items
inserted into the reservoir times the time to insert them.

\begin{corollary}\label{cor:u}
  For unweighted (uniform) reservoir sampling with sample size $k$, processing a
  mini-batch is possible in $\Oh{(b^*+1) \log(b^*+k) + \Tsel}$ time,
  where $b^* \leq b$ is the maximum number of items from the mini-batch on any
  PE that is below the insertion threshold, and $\Tsel$ is the time for
  selection from sorted sequences of size at most $b^*+k$ per PE (see
  \cref{prel:sel}).
\end{corollary}

Observe that the criteria for random distribution of \cref{prel:sel:rand} are
much easier to satisfy for uniform sampling, resulting in only an
$\Tstart \log p$ contribution for the $\Tsel$ term in the running time bound in
these cases.  For example, a uniform arrival rate for all PEs suffices, as the
keys associated with the items are uniformly random.

Refer to Lemma~\ref{lem:load} and \cref{thm:load} above for an analysis of the
number of items below the insertion threshold over all mini-batches.

\subsection{Reservoir Sampling with Variable Reservoir Size}\label{res:var}

For any given threshold $T$, the items with keys less than or equal to $T$ form a
sample without replacement of all items seen so far.  The size of this sample --
call it $s$ -- is not known a priori, and in the previous sections, we used
selection from the local reservoirs to determine the threshold so that $s = k$.
If, however, the precise sample size is not important to the application,
and $s$ is allowed to vary in some range $\kmin..\kmax$, we can do better than
for fixed $k$.  By using the \emph{amsSelect} approximate selection algorithm
\cite[Section IV-C]{HubSan16topk} (see also \cref{prel:sel:approx}), selection
converges much faster if $\kmax-\kmin$ is sufficiently large (\eg a constant
factor apart).

\begin{corollary}
  If the sample size is allowed to vary between $\kmin$ and $\kmax$, and
  $\kmax - \kmin = \Om{\kmax}$, then $\Tsel = \Tstart \log p$ in the running
  time of \cref{thm:res} and Corollary~\ref{cor:u}.
\end{corollary}

Observe that if the items come from a common distribution, once $n \gg k$ items
have been processed, turnover in the sample is very low.  Accordingly, only few
items have keys below the threshold to enter the reservoir in each batch.  As a
result, we can forego selection and let the sample grow until $s>\kmax$,
potentially for multiple mini-batches.  Only then does the selection have to find a new threshold.
Additionally, the selection is faster because it does not have to find the item
with a particular precise rank, but only \emph{some} item in a given range of
ranks.

\begin{figure}[t]
  \centering
  \includegraphics{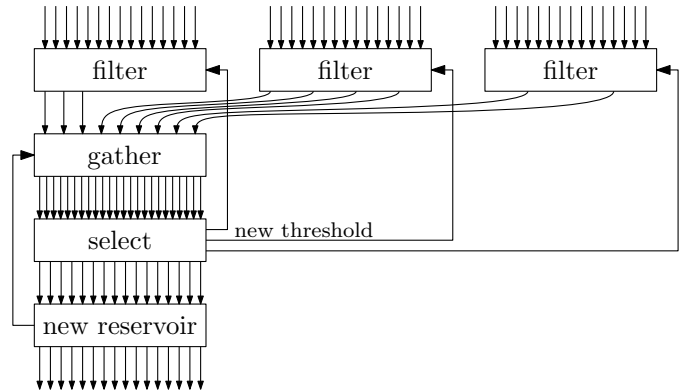}
  \caption{Schematic view of data flow in the centralized gathering
    algorithm of \cref{res:gather}.}\label{fig:schema_gather}
\end{figure}

\subsection{A Centralized Approach}\label{res:gather}

As a comparison point to our fully decentralized algorithms, we now describe an
approach that designates a single PE as the coordinator, which maintains the
sample.  Again, the PEs use a threshold to immediately discard any items that
cannot be part of the sample (in the first batch, if a PE receives more than $k$
items, only the $k$ items with the smallest keys need to be retained).  The
remaining candidates are gathered at the coordinator, which uses a standard
sequential selection algorithm (\eg quickselect) to sequentially select the $k$
smallest items, discards the rest, and broadcasts the new threshold.
\Cref{fig:schema_gather} shows a schematic view of the algorithm.

Observe that the number of items gathered in later batches is small in expectation, as only few
new items enter the reservoir if the items follow the same distribution in all
mini-batches.  This algorithm can be seen as an adaptation of Jayaram \etal's
method~\cite{JSTW19wres} to a mini-batch model, which renders the level set
construction used therein unnecessary.

\section{Implementation Optimizations}\label{impl}

We now discuss two optimizations that make implementations of our
algorithm more efficient.

To avoid inserting too many items into the reservoir during the first mini-batch
if its local size $b$ is large compared to the sample size $k$, we use the key
of \emph{local rank} $k$, which is periodically re-determined, as a local
threshold in the first batch.  More concretely, if $b \geq \max(1.5k, k+500)$,
we use the key of \emph{local rank} $k$ as threshold for subsequent items, and
refresh this local threshold every time the local reservoir exceeds
$\max(1.1k, k+250)$ items, discarding those that are larger.  This maintains
correctness: at no point is a local reservoir pruned to a size smaller than $k$,
so the union of all local reservoirs is guaranteed to be of size at least $k$.
It also maintains the property that each local reservoir is a sample without
replacement of some size $k'$ of all items seen so far this PE.

Furthermore, to speed up the innermost loop of Algorithm~\ref{alg:w}, we compute
the sum of weights of 32 items at a time, check whether the amount of weight to
be skipped ($X$) is larger than this sum, and skip all 32 items at once if this is the case.
This reduces the number of conditional branches and allows for vectorization
using the CPU's Single Instruction Multiple Data (SIMD) units.  Both of these
factors speed up processing of the items in a batch significantly -- especially
as typically, only few items per batch end up contributing to the reservoir (and $X$ is much
larger than the average weight) once $n\gg k$ items have been processed.

\section{Experiments}\label{exp}

We implemented our weighted reservoir sampling algorithm of \cref{res:w} with
the approximate selection algorithm of \cref{prel:sel:approx}, using one
(labeled \emph{``ours''}) or several pivots (\emph{``ours-8''} for 8 pivots,
chosen in a preliminary experiment) and exact bounds (\ie
$\kmin=\kmax=k$)\footnote{In this configuration, its asymptotical running time
  matches the non-approximate algorithm \cite[Section IV-B]{HubSan16topk}, but
  the approximate algorithm's more advanced pivot selection speeds up
  convergence in practice.}, as well as the centralized gathering algorithm of
\cref{res:gather}.  Here, we present the results of strong and weak scaling
experiments on a supercomputer.  Our evaluation is structured as follows.
First, we describe the setup and implementation details in \cref{exp:setup}.
The strong and weak scaling evaluations follow in \cref{exp:strong,exp:weak}.
Lastly, \cref{exp:composition} looks at the composition of the running time of
the algorithms.

Throughout this section, our weighted reservoir sampling algorithm with
single-pivot selection is referred to as \emph{ours}, its version with
multi-pivot selection with 8 pivots as \emph{ours-8}, and the centralized
gathering algorithm as \emph{gather}.

\subsection{Experimental Setup}\label{exp:setup}

We used C++17 for the implementation and MPI for communication between the
PEs.  The code was compiled with the GNU C++ compiler \texttt{g++} in version 8.3.0 with
optimization flags \texttt{-O3 -march=native} and run with OpenMPI 4.0.
The experiments were run on ForHLR~II, a general-purpose high-performance computing
cluster located at Karlsruhe Institute of Technology, using up to 256 nodes.
Each node is equipped with two deca-core Intel Xeon E5-2660 v3 processors for a
total of 20 cores per node, and 64\,GiB of DDR4 main memory.  We use one MPI
process (PE) per core, \ie 20 PEs per node, for a total of up to 5120 PEs.  All nodes
are attached to an InfiniBand 4X EDR interconnect using an InfiniBand FDR
adapter \cite{forhlr2}.  Our implementation is licensed under the GPLv3 and
available at \break\url{https://github.com/lorenzhs/reservoir}.

We use Intel's Math Kernel Library 2019~\cite{intel-mkl} for fast random number
generation using a Mersenne Twister~\cite{MatNis98}.  The local reservoirs are
implemented as B+ trees, based on an implementation of Bingmann~\cite{tlx}
augmented with join/split operations (see, \eg \cite{MehSanPar}) from an
implementation by Akhremtsev~\cite{akhremtsev2016tree} and rank/select support
(\eg \cite{MehSanPar}).

Each experiment was run ten times, and each run lasted 30 seconds, completing as
many mini-batches as possible in that time.  Input generation is not included in
the reported times.  We use uniformly random floating-point weights from the
range 0..100 as inputs.  Preliminary experiments with skewed weights -- normally
distributed with the mean increasing based on the mini-batch sequence number and the PEs' ranks --
showed no significant differences in running time.  Speedups are reported relative
to our algorithm with single-pivot selection \emph{(``ours'')} on 1 node ($p=20$ cores).
Based on preliminary experiments, we chose $d=8$ as the number of pivots used in
the multi-pivot selection algorithm of \cref{prel:sel:general}.

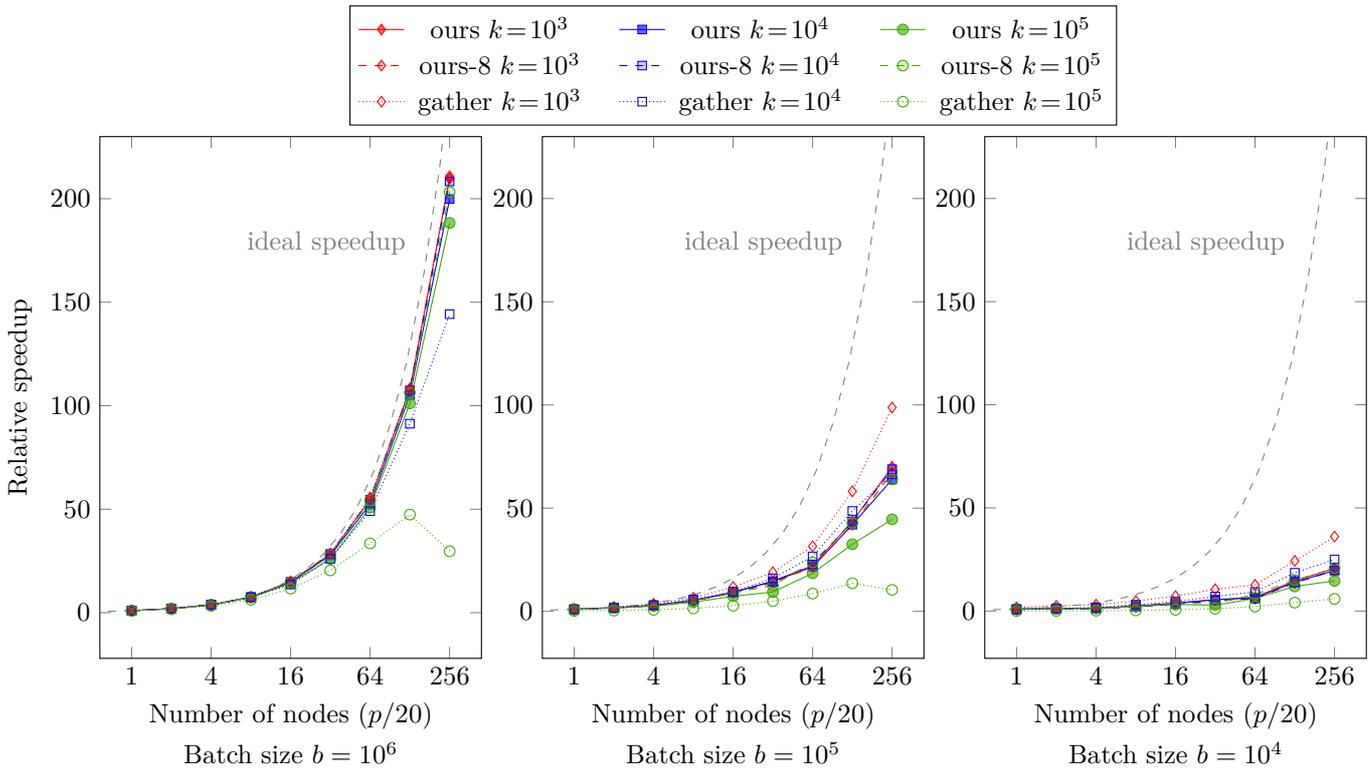
\begin{figure*}[t]
\hspace*{-6.5em}
\begin{tikzpicture}
  \pgfplotsset{every axis title/.append style={at={(0.5,-0.25)}}}
  \begin{groupplot}[
    group style = {group size = 3 by 1, horizontal sep = 8mm},
    xmode=log,
    log basis x=2,
    log ticks with fixed point,
    xlabel={Number of nodes ($p/20$)},
    legend pos=north west,
    width=6.6cm, %
    height=8.5cm,
    domain=0.5:512,
    cycle list name=mylist,
    ymax=230,
  ]
\nextgroupplot[title={Batch size $b=10^6$}, ylabel={Relative speedup},
legend style = {
  column sep=3pt,
  legend columns = 3,
  legend to name = legendweak,
  /tikz/column 2/.style={column sep=12pt},
  /tikz/column 4/.style={column sep=12pt}
}]
  \addplot coordinates { (1,1.0) (2,1.97198) (4,3.8021) (8,7.52469) (16,14.8619) (32,28.2409) (64,54.5395) (128,107.841) (256,209.805) };
  \addlegendentry{ours $k\!=\!10^3$};
  \addplot coordinates { (1,1.0) (2,1.96585) (4,3.77906) (8,7.41025) (16,14.6525) (32,27.7347) (64,52.4111) (128,105.096) (256,199.84) };
  \addlegendentry{ours $k\!=\!10^4$};
  \addplot coordinates { (1,1.0) (2,1.93612) (4,3.74939) (8,7.30577) (16,14.2286) (32,25.9694) (64,50.7726) (128,101.124) (256,188.28) };
  \addlegendentry{ours $k\!=\!10^5$};
  \addplot coordinates { (1,0.999468) (2,1.97279) (4,3.81477) (8,7.53796) (16,14.9065) (32,28.313) (64,54.9338) (128,107.98) (256,209.943) };
  \addlegendentry{ours-8 $k\!=\!10^3$};
  \addplot coordinates { (1,1.00002) (2,1.976) (4,3.8297) (8,7.53071) (16,14.9436) (32,28.3877) (64,54.5379) (128,107.401) (256,208.325) };
  \addlegendentry{ours-8 $k\!=\!10^4$};
  \addplot coordinates { (1,1.00217) (2,1.95544) (4,3.83838) (8,7.50476) (16,14.7444) (32,27.6536) (64,54.2219) (128,106.321) (256,203.229) };
  \addlegendentry{ours-8 $k\!=\!10^5$};
  \addplot coordinates { (1,0.99834) (2,1.98156) (4,3.8406) (8,7.60387) (16,15.0625) (32,28.6288) (64,55.6671) (128,108.428) (256,210.879) };
  \addlegendentry{gather $k\!=\!10^3$};
  \addplot coordinates { (1,0.987101) (2,1.96349) (4,3.82387) (8,7.49477) (16,13.6136) (32,26.0629) (64,49.0622) (128,91.2965) (256,144.233) };
  \addlegendentry{gather $k\!=\!10^4$};
  \addplot coordinates { (1,0.849084) (2,1.66449) (4,3.24181) (8,6.21844) (16,11.6459) (32,20.4279) (64,33.4556) (128,47.3937) (256,29.6756) };
  \addlegendentry{gather $k\!=\!10^5$};

  \addplot[overlay,gray,dashed,update limits=false]{x} node[left,pos=0.35] {ideal speedup\kern0.5em};
  \nextgroupplot[title={Batch size $b=10^5$}]
  \addplot coordinates { (1,1.0) (2,1.78431) (4,2.87778) (8,5.22937) (16,9.21663) (32,14.6174) (64,22.1568) (128,43.509) (256,69.9572) };
  \addlegendentry{ours $k\!=\!10^3$};
  \addplot coordinates { (1,1.0) (2,1.71827) (4,2.84903) (8,5.24155) (16,9.16976) (32,14.5124) (64,21.4919) (128,41.964) (256,63.8581) };
  \addlegendentry{ours $k\!=\!10^4$};
  \addplot coordinates { (1,1.0) (2,1.58761) (4,2.52302) (8,4.44489) (16,7.30533) (32,9.3584) (64,18.4219) (128,32.5041) (256,44.5561) };
  \addlegendentry{ours $k\!=\!10^5$};
  \addplot coordinates { (1,0.953833) (2,1.67373) (4,2.74447) (8,4.9593) (16,8.78725) (32,13.8655) (64,21.5389) (128,41.9512) (256,66.9647) };
  \addlegendentry{ours-8 $k\!=\!10^3$};
  \addplot coordinates { (1,0.957578) (2,1.64433) (4,2.82218) (8,5.18607) (16,9.17889) (32,14.2668) (64,22.7132) (128,43.5674) (256,68.8819) };
  \addlegendentry{ours-8 $k\!=\!10^4$};
  \addplot coordinates { (1,1.00228) (2,1.65824) (4,2.86406) (8,5.16233) (16,8.96895) (32,12.521) (64,23.7092) (128,42.8687) (256,63.8969) };
  \addlegendentry{ours-8 $k\!=\!10^5$};
  \addplot coordinates { (1,1.07382) (2,2.09115) (4,3.63459) (8,6.56122) (16,11.8178) (32,18.9268) (64,31.5051) (128,58.1536) (256,98.7853) };
  \addlegendentry{gather $k\!=\!10^3$};
  \addplot coordinates { (1,0.889663) (2,1.67046) (4,3.05547) (8,5.78962) (16,9.54066) (32,15.9876) (64,26.6101) (128,48.6192) (256,66.3357) };
  \addlegendentry{gather $k\!=\!10^4$};
  \addplot coordinates { (1,0.193007) (2,0.378354) (4,0.734243) (8,1.43499) (16,2.70257) (32,4.94745) (64,8.61951) (128,13.604) (256,10.3948) };
  \addlegendentry{gather $k\!=\!10^5$};

  \legend{};
  \addplot[overlay,gray,dashed,update limits=false]{x} node[left,pos=0.35] {ideal speedup\kern0.7em};
  \nextgroupplot[title={Batch size $b=10^4$}]
  \addplot coordinates { (1,1.0) (2,1.24267) (4,1.34435) (8,2.39364) (16,3.79805) (32,5.53814) (64,6.70825) (128,14.6383) (256,20.6978) };
  \addlegendentry{ours $k\!=\!10^3$};
  \addplot coordinates { (1,1.0) (2,1.23576) (4,1.32363) (8,2.36054) (16,3.76536) (32,5.47028) (64,6.31521) (128,14.0007) (256,19.6713) };
  \addlegendentry{ours $k\!=\!10^4$};
  \addplot coordinates { (1,1.0) (2,1.20127) (4,1.30264) (8,2.18618) (16,3.1683) (32,2.81778) (64,6.12075) (128,11.9266) (256,14.6114) };
  \addlegendentry{ours $k\!=\!10^5$};
  \addplot coordinates { (1,0.791845) (2,0.974322) (4,1.181) (8,2.1046) (16,3.38245) (32,5.0971) (64,6.09704) (128,13.6151) (256,19.5769) };
  \addlegendentry{ours-8 $k\!=\!10^3$};
  \addplot coordinates { (1,0.792963) (2,0.969364) (4,1.1945) (8,2.12128) (16,3.44681) (32,5.13962) (64,6.05143) (128,13.8685) (256,19.7199) };
  \addlegendentry{ours-8 $k\!=\!10^4$};
  \addplot coordinates { (1,0.857595) (2,1.05025) (4,1.35942) (8,2.32632) (16,3.56575) (32,3.81579) (64,8.01086) (128,14.9518) (256,20.7404) };
  \addlegendentry{ours-8 $k\!=\!10^5$};
  \addplot coordinates { (1,1.73501) (2,2.52786) (4,3.04312) (8,4.8478) (16,7.23034) (32,10.5405) (64,12.6545) (128,24.2898) (256,36.0793) };
  \addlegendentry{gather $k\!=\!10^3$};
  \addplot coordinates { (1,0.679059) (2,1.15693) (4,1.76064) (8,3.07995) (16,4.52755) (32,7.05437) (64,9.26404) (128,18.54) (256,25.0178) };
  \addlegendentry{gather $k\!=\!10^4$};
  \addplot coordinates { (1,0.0411375) (2,0.08094) (4,0.157211) (8,0.31199) (16,0.592804) (32,1.10714) (64,2.14964) (128,4.05801) (256,5.87189) };
  \addlegendentry{gather $k\!=\!10^5$};

  \legend{}
  \addplot[overlay,gray,dashed,update limits=false]{x} node[left,pos=0.35] {ideal speedup\kern0.7em};
\end{groupplot}

\hidelinks
\node at ($(group c2r1) + (0,4.4)$) {\ref{legendweak}};
\restorelinks

\end{tikzpicture}
\caption{Weak scaling with different batch and sample sizes. Speedups are
  relative to our algorithm with single-pivot selection (\emph{ours}) for the same sample size
  on 1 node (20 cores).}
\label{fig:weak}
\end{figure*}

\subsection{Weak Scaling}\label{exp:weak}

The results of our weak scaling experiments are shown in \cref{fig:weak}, with
three plots for per-PE mini-batch sizes $b$ -- from left to right -- $10^6$,
$10^5$, and $10^4$, each with sample sizes $k$ of $10^3$, $10^4$, and $10^5$
items.  Speedups are shown relative to our algorithm (\emph{ours}) for the same
sample size on a single node ($p=20$).

We can see that our algorithm shows good scaling across the board, and, as
expected, smaller samples achieve slightly better speedups than larger ones.  The
causes for this lie in the $\Ohsmash{\log^2(kp)}$ latency of the selection
algorithm and increased local processing time due to larger local reservoirs.
Using multiple pivots for selection (\emph{ours-8}, semi-filled marks with
dashed lines) is especially beneficial for larger sample sizes (in blue and green), where
it reduces average recursion depth by a factor of around 2.5 -- from 7.3 to 2.7
for $k=10^5$ and from 4.3 to 1.8 for $k=10^4$ -- compared to a much smaller
improvement from 1.9 to 1.1 for $k=10^3$ (in red), where the average recursion depth is
already very low when using a single pivot.  This results in selection running
time improvements of up to $35\,\%$ for $k=10^5$ and around $20\,\%$ for $k=10^4$, with
no significant improvement for $k=10^3$.  Because local processing is a
significant part of overall processing time (refer to the running time
composition analysis of \cref{exp:composition} for details), the actual overall
running time improvement is only around $8\,\%$ ($k=10^5$ and $b=10^6$).

We also see clearly that the centralized algorithm (\emph{gather}, hollow marks
with dotted lines) performs well only when the sample size is very small
($k=10^3$, in red), where few candidates need to be gathered per batch.  It
begins struggling even with $k=10^4$ for large batch sizes.  For larger sample
sizes ($k=10^5$), it performs badly regardless of batch size.  All algorithms
achieve better---and in the case of our algorithm, near-optimal---speedups for
large batch sizes, as communication overhead is much less noticeable than for
small batches, where local processing is fast.

Overall, \emph{ours-8} is the most consistently fast algorithm, and should be
preferred over the version with single-pivot selection (\emph{ours}).

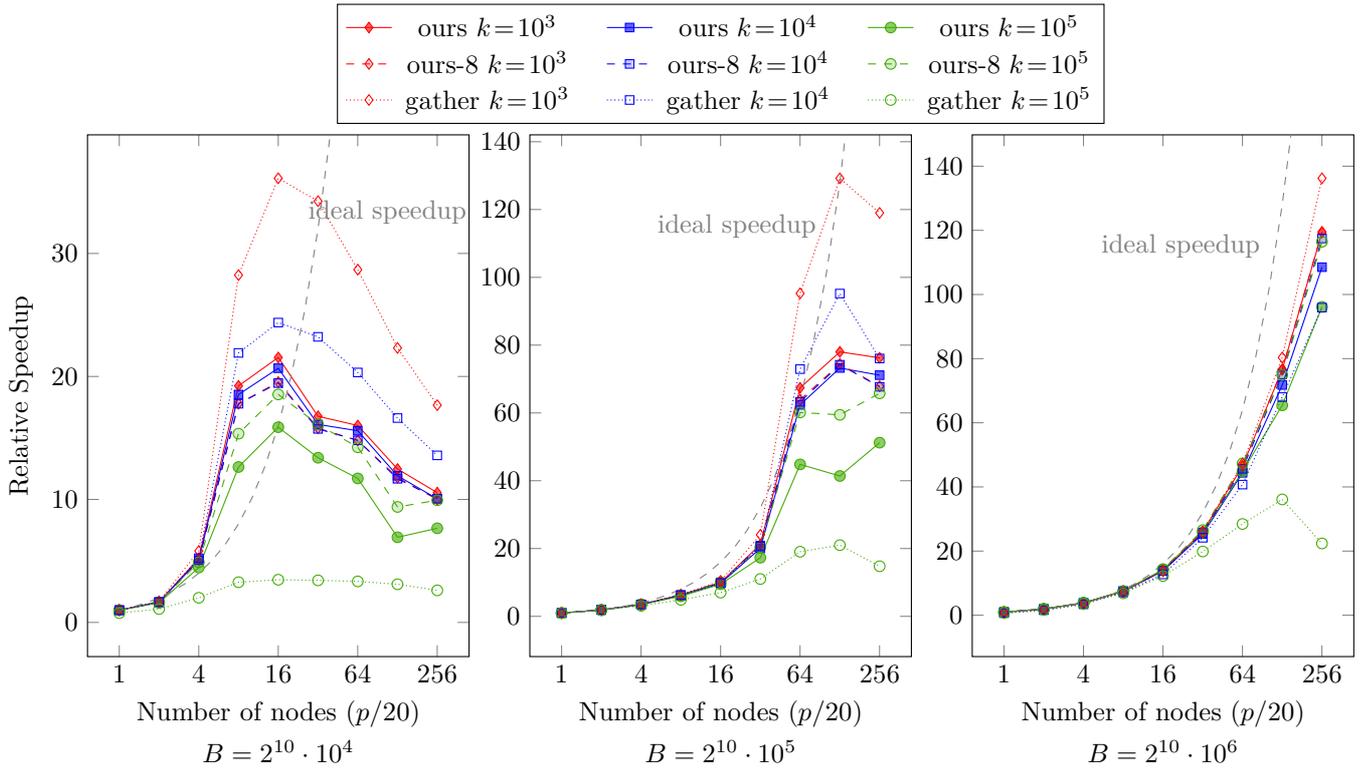
\begin{figure*}[p]
\hspace*{-5.4em}
\begin{tikzpicture}
  \pgfplotsset{every axis title/.append style={at={(0.5,-0.25)}}}
  \begin{groupplot}[
    group style = {group size = 3 by 1, horizontal sep = 8mm},
    xmode=log,
    log basis x=2,
    log ticks with fixed point,
    xlabel={Number of nodes ($p/20$)},
    legend pos=north west,
    width=6.6cm, %
    height=8.5cm,
    domain=1:256,
    cycle list name=mylist,
  ]
\nextgroupplot[title = {$B=2^{10}\cdot 10^4$}, ylabel={Relative Speedup},
legend style = {
  column sep=3pt,
  legend columns = 3,
  legend to name = legendstrong,
  /tikz/column 2/.style={column sep=12pt},
  /tikz/column 4/.style={column sep=12pt}
}]
  \addplot coordinates { (1,1.0) (2,1.65933) (4,5.211) (8,19.2334) (16,21.5412) (32,16.7619) (64,16.0107) (128,12.4737) (256,10.5509) };
  \addlegendentry{ours $k\!=\!10^3$};
  \addplot coordinates { (1,1.0) (2,1.64421) (4,5.05682) (8,18.5291) (16,20.6652) (32,16.0981) (64,15.5834) (128,11.9221) (256,10.047) };
  \addlegendentry{ours $k\!=\!10^4$};
  \addplot coordinates { (1,1.0) (2,1.59601) (4,4.46095) (8,12.6367) (16,15.874) (32,13.3887) (64,11.7092) (128,6.92324) (256,7.65443) };
  \addlegendentry{ours $k\!=\!10^5$};
  \addplot coordinates { (1,0.996752) (2,1.64603) (4,5.08871) (8,17.7864) (16,19.5403) (32,15.7892) (64,14.7891) (128,11.7839) (256,9.99564) };
  \addlegendentry{ours-8 $k\!=\!10^3$};
  \addplot coordinates { (1,1.00032) (2,1.65127) (4,5.07336) (8,17.7814) (16,19.4609) (32,15.7348) (64,14.8059) (128,11.6768) (256,10.0642) };
  \addlegendentry{ours-8 $k\!=\!10^4$};
  \addplot coordinates { (1,1.00304) (2,1.63673) (4,4.89235) (8,15.3462) (16,18.5366) (32,16.0981) (64,14.2429) (128,9.39011) (256,9.93053) };
  \addlegendentry{ours-8 $k\!=\!10^5$};
  \addplot coordinates { (1,1.00495) (2,1.71299) (4,5.81184) (8,28.2409) (16,36.1083) (32,34.2519) (64,28.6864) (128,22.3183) (256,17.6656) };
  \addlegendentry{gather $k\!=\!10^3$};
  \addplot coordinates { (1,0.984594) (2,1.65476) (4,5.21865) (8,21.9127) (16,24.3797) (32,23.2193) (64,20.3225) (128,16.6172) (256,13.5828) };
  \addlegendentry{gather $k\!=\!10^4$};
  \addplot coordinates { (1,0.758106) (2,1.08791) (4,2.0025) (8,3.26877) (16,3.47225) (32,3.42022) (64,3.33593) (128,3.10466) (256,2.60774) };
  \addlegendentry{gather $k\!=\!10^5$};

  \addplot[overlay,gray,dashed,update limits=false]{x} node[right,pos=0.13] {\kern-0.8em ideal speedup};

\nextgroupplot[title = {$B=2^{10}\cdot 10^5$}]
  \addplot coordinates { (1,1.0) (2,1.90876) (4,3.52324) (8,6.29956) (16,10.0089) (32,21.1932) (64,67.405) (128,78.0135) (256,76.2506) };
  \addlegendentry{ours $k\!=\!10^3$};
  \addplot coordinates { (1,1.0) (2,1.89662) (4,3.48782) (8,6.16562) (16,9.7454) (32,20.0575) (64,62.3827) (128,73.2555) (256,71.1319) };
  \addlegendentry{ours $k\!=\!10^4$};
  \addplot coordinates { (1,1.0) (2,1.90685) (4,3.47712) (8,5.91729) (16,9.45981) (32,17.3036) (64,44.8233) (128,41.4006) (256,51.2369) };
  \addlegendentry{ours $k\!=\!10^5$};
  \addplot coordinates { (1,0.998822) (2,1.9116) (4,3.53133) (8,6.31307) (16,9.99748) (32,20.8983) (64,63.8609) (128,74.4162) (256,67.8839) };
  \addlegendentry{ours-8 $k\!=\!10^3$};
  \addplot coordinates { (1,1.00224) (2,1.90433) (4,3.52803) (8,6.27063) (16,9.97345) (32,20.8315) (64,63.2854) (128,74.2132) (256,67.7157) };
  \addlegendentry{ours-8 $k\!=\!10^4$};
  \addplot coordinates { (1,1.00059) (2,1.91641) (4,3.54113) (8,6.17515) (16,10.0583) (32,20.3797) (64,60.1432) (128,59.4413) (256,65.7782) };
  \addlegendentry{ours-8 $k\!=\!10^5$};
  \addplot coordinates { (1,0.956) (2,1.87672) (4,3.5249) (8,6.4147) (16,10.4527) (32,24.0419) (64,95.2504) (128,129.205) (256,118.996) };
  \addlegendentry{gather $k\!=\!10^3$};
  \addplot coordinates { (1,0.955721) (2,1.86214) (4,3.49402) (8,6.25705) (16,9.8213) (32,20.6064) (64,72.9266) (128,95.219) (256,76.0595) };
  \addlegendentry{gather $k\!=\!10^4$};
  \addplot coordinates { (1,0.959167) (2,1.78332) (4,3.10323) (8,4.8659) (16,7.00749) (32,11.0185) (64,19.0596) (128,21.003) (256,14.7797) };
  \addlegendentry{gather $k\!=\!10^5$};

  \legend{} %
  \addplot[overlay,gray,dashed,update limits=false]{x} node[left,pos=0.45] {ideal speedup\kern0.3em};

\nextgroupplot[title = {$B=2^{10}\cdot 10^6$}]
  \addplot coordinates { (1,1.0) (2,1.9656) (4,3.85999) (8,7.54123) (16,14.0942) (32,26.0564) (64,46.4773) (128,76.5713) (256,119.33) };
  \addlegendentry{ours $k\!=\!10^3$};
  \addplot coordinates { (1,1.0) (2,1.93305) (4,3.81295) (8,7.49019) (16,13.7861) (32,25.3309) (64,44.3815) (128,71.7778) (256,108.517) };
  \addlegendentry{ours $k\!=\!10^4$};
  \addplot coordinates { (1,1.0) (2,1.96854) (4,3.89196) (8,7.44938) (16,14.2407) (32,25.8157) (64,44.8131) (128,65.4422) (256,96.1365) };
  \addlegendentry{ours $k\!=\!10^5$};
  \addplot coordinates { (1,1.00399) (2,1.97682) (4,3.83236) (8,7.49941) (16,14.0352) (32,26.1003) (64,46.5995) (128,76.6444) (256,119.527) };
  \addlegendentry{ours-8 $k\!=\!10^3$};
  \addplot coordinates { (1,0.999027) (2,1.9445) (4,3.80693) (8,7.51255) (16,13.8728) (32,25.7196) (64,45.6497) (128,75.073) (256,117.452) };
  \addlegendentry{ours-8 $k\!=\!10^4$};
  \addplot coordinates { (1,1.00046) (2,1.98986) (4,3.906) (8,7.51481) (16,14.3889) (32,26.5315) (64,47.3327) (128,74.9248) (256,116.436) };
  \addlegendentry{ours-8 $k\!=\!10^5$};
  \addplot coordinates { (1,0.658675) (2,1.58548) (4,3.45851) (8,7.15161) (16,13.6948) (32,25.896) (64,47.2987) (128,80.3549) (256,136.215) };
  \addlegendentry{gather $k\!=\!10^3$};
  \addplot coordinates { (1,0.663312) (2,1.57502) (4,3.40442) (8,7.10324) (16,12.7024) (32,24.144) (64,40.7213) (128,68.0176) (256,95.9063) };
  \addlegendentry{gather $k\!=\!10^4$};
  \addplot coordinates { (1,0.710004) (2,1.67301) (4,3.53722) (8,6.83042) (16,12.2319) (32,19.8704) (64,28.4434) (128,36.1031) (256,22.3728) };
  \addlegendentry{gather $k\!=\!10^5$};

  \legend{} %
  \addplot[overlay,gray,dashed,update limits=false]{x} node[left,pos=0.45] {ideal speedup\kern0.25em};
\end{groupplot}

\hidelinks
\node at ($(group c2r1) + (0,4.4)$) {\ref{legendstrong}};
\restorelinks
\end{tikzpicture}
\caption{Strong scaling, speedups relative to our algorithm with single-pivot
  selection (\emph{ours}) for the same sample size on 1 node (20 cores).}
\label{fig:strong_speedup}
\end{figure*}

\begin{figure*}[p]
\hspace*{-5.4em}
\begin{tikzpicture}
  \pgfplotsset{every axis title/.append style={at={(0.5,-0.25)}}}
  \begin{groupplot}[
    group style = {group size = 3 by 1, horizontal sep = 8mm},
    xmode=log,
    log basis x=2,
    log ticks with fixed point,
    xlabel={Number of nodes ($p/20$)},
    legend pos=north west,
    width=6.6cm, %
    height=8.5cm,
    domain=1:256,
    cycle list name=mylist,
  ]
\nextgroupplot[title = {$B=2^{10}\cdot 10^4$}, ylabel={Throughput per PE (items/s)},
]
  \addplot coordinates { (1,2.82371e+08) (2,2.34273e+08) (4,3.67858e+08) (8,6.78869e+08) (16,3.80163e+08) (32,1.47909e+08) (64,7.064e+07) (128,2.75173e+07) (256,1.16378e+07) };
  \addlegendentry{selection=ams-select,sample\_size=1000};
  \addplot coordinates { (1,2.80159e+08) (2,2.30321e+08) (4,3.54179e+08) (8,6.4889e+08) (16,3.61848e+08) (32,1.40939e+08) (64,6.82164e+07) (128,2.60944e+07) (256,1.09951e+07) };
  \addlegendentry{selection=ams-select,sample\_size=10000};
  \addplot coordinates { (1,2.73121e+08) (2,2.17953e+08) (4,3.04595e+08) (8,4.3142e+08) (16,2.7097e+08) (32,1.14273e+08) (64,4.99692e+07) (128,1.47725e+07) (256,8.16637e+06) };
  \addlegendentry{selection=ams-select,sample\_size=100000};
  \addplot coordinates { (1,2.81454e+08) (2,2.32396e+08) (4,3.59226e+08) (8,6.27792e+08) (16,3.44849e+08) (32,1.39325e+08) (64,6.52502e+07) (128,2.59955e+07) (256,1.10253e+07) };
  \addlegendentry{selection=ams-multi-8,sample\_size=1000};
  \addplot coordinates { (1,2.8025e+08) (2,2.31309e+08) (4,3.55337e+08) (8,6.22702e+08) (16,3.40759e+08) (32,1.37758e+08) (64,6.48125e+07) (128,2.55575e+07) (256,1.1014e+07) };
  \addlegendentry{selection=ams-multi-8,sample\_size=10000};
  \addplot coordinates { (1,2.73951e+08) (2,2.23514e+08) (4,3.34051e+08) (8,5.23923e+08) (16,3.16422e+08) (32,1.37398e+08) (64,6.07819e+07) (128,2.00362e+07) (256,1.05947e+07) };
  \addlegendentry{selection=ams-multi-8,sample\_size=100000};
  \addplot coordinates { (1,2.83768e+08) (2,2.4185e+08) (4,4.10274e+08) (8,9.96803e+08) (16,6.37246e+08) (32,3.02242e+08) (64,1.26565e+08) (128,4.92346e+07) (256,1.94854e+07) };
  \addlegendentry{selection=gather,sample\_size=1000};
  \addplot coordinates { (1,2.75843e+08) (2,2.31799e+08) (4,3.65514e+08) (8,7.67379e+08) (16,4.26888e+08) (32,2.03285e+08) (64,8.89616e+07) (128,3.63707e+07) (256,1.48646e+07) };
  \addlegendentry{selection=gather,sample\_size=10000};
  \addplot coordinates { (1,2.07055e+08) (2,1.48566e+08) (4,1.36731e+08) (8,1.11596e+08) (16,5.92716e+07) (32,2.91917e+07) (64,1.42361e+07) (128,6.62463e+06) (256,2.78215e+06) };
  \addlegendentry{selection=gather,sample\_size=100000};

  \legend{} %
\nextgroupplot[title = {$B=2^{10}\cdot 10^5$}]
  \addplot coordinates { (1,3.61083e+08) (2,3.4461e+08) (4,3.18045e+08) (8,2.84332e+08) (16,2.25878e+08) (32,2.3914e+08) (64,3.80293e+08) (128,2.20073e+08) (256,1.0755e+08) };
  \addlegendentry{selection=ams-select,sample\_size=1000};
  \addplot coordinates { (1,3.58899e+08) (2,3.40349e+08) (4,3.12944e+08) (8,2.76604e+08) (16,2.18601e+08) (32,2.24957e+08) (64,3.49829e+08) (128,2.05401e+08) (256,9.97232e+07) };
  \addlegendentry{selection=ams-select,sample\_size=10000};
  \addplot coordinates { (1,3.43089e+08) (2,3.27111e+08) (4,2.98241e+08) (8,2.5377e+08) (16,2.02848e+08) (32,1.85522e+08) (64,2.40288e+08) (128,1.10969e+08) (256,6.86673e+07) };
  \addlegendentry{selection=ams-select,sample\_size=100000};
  \addplot coordinates { (1,3.60657e+08) (2,3.45124e+08) (4,3.18775e+08) (8,2.84943e+08) (16,2.2562e+08) (32,2.35812e+08) (64,3.60297e+08) (128,2.09925e+08) (256,9.57489e+07) };
  \addlegendentry{selection=ams-multi-8,sample\_size=1000};
  \addplot coordinates { (1,3.59704e+08) (2,3.4173e+08) (4,3.16552e+08) (8,2.81315e+08) (16,2.23716e+08) (32,2.33637e+08) (64,3.54891e+08) (128,2.08086e+08) (256,9.49337e+07) };
  \addlegendentry{selection=ams-multi-8,sample\_size=10000};
  \addplot coordinates { (1,3.43293e+08) (2,3.28749e+08) (4,3.03731e+08) (8,2.64829e+08) (16,2.15682e+08) (32,2.18501e+08) (64,3.22414e+08) (128,1.59326e+08) (256,8.81554e+07) };
  \addlegendentry{selection=ams-multi-8,sample\_size=100000};
  \addplot coordinates { (1,3.45195e+08) (2,3.38826e+08) (4,3.18196e+08) (8,2.89529e+08) (16,2.35894e+08) (32,2.71284e+08) (64,5.37394e+08) (128,3.64482e+08) (256,1.67842e+08) };
  \addlegendentry{selection=gather,sample\_size=1000};
  \addplot coordinates { (1,3.43007e+08) (2,3.3416e+08) (4,3.135e+08) (8,2.80706e+08) (16,2.20303e+08) (32,2.31112e+08) (64,4.08957e+08) (128,2.66984e+08) (256,1.06631e+08) };
  \addlegendentry{selection=gather,sample\_size=10000};
  \addplot coordinates { (1,3.2908e+08) (2,3.0592e+08) (4,2.66171e+08) (8,2.0868e+08) (16,1.50262e+08) (32,1.18135e+08) (64,1.02175e+08) (128,5.6296e+07) (256,1.98076e+07) };
  \addlegendentry{selection=gather,sample\_size=100000};

  \legend{} %

\nextgroupplot[title = {$B=2^{10}\cdot 10^6$}]
  \addplot coordinates { (1,3.69311e+08) (2,3.62957e+08) (4,3.56384e+08) (8,3.48132e+08) (16,3.25323e+08) (32,3.00716e+08) (64,2.68196e+08) (128,2.20927e+08) (256,1.72147e+08) };
  \addlegendentry{selection=ams-select,sample\_size=1000};
  \addplot coordinates { (1,3.70523e+08) (2,3.58119e+08) (4,3.53196e+08) (8,3.46911e+08) (16,3.19252e+08) (32,2.93302e+08) (64,2.56943e+08) (128,2.07776e+08) (256,1.57062e+08) };
  \addlegendentry{selection=ams-select,sample\_size=10000};
  \addplot coordinates { (1,3.4325e+08) (2,3.37851e+08) (4,3.33979e+08) (8,3.19625e+08) (16,3.05508e+08) (32,2.76914e+08) (64,2.40345e+08) (128,1.75493e+08) (256,1.28902e+08) };
  \addlegendentry{selection=ams-select,sample\_size=100000};
  \addplot coordinates { (1,3.70784e+08) (2,3.65031e+08) (4,3.53832e+08) (8,3.46201e+08) (16,3.2396e+08) (32,3.01223e+08) (64,2.68901e+08) (128,2.21138e+08) (256,1.72433e+08) };
  \addlegendentry{selection=ams-multi-8,sample\_size=1000};
  \addplot coordinates { (1,3.70162e+08) (2,3.60239e+08) (4,3.52639e+08) (8,3.47946e+08) (16,3.21263e+08) (32,2.97803e+08) (64,2.64286e+08) (128,2.17314e+08) (256,1.69995e+08) };
  \addlegendentry{selection=ams-multi-8,sample\_size=10000};
  \addplot coordinates { (1,3.4341e+08) (2,3.41509e+08) (4,3.35184e+08) (8,3.22433e+08) (16,3.08687e+08) (32,2.84592e+08) (64,2.53859e+08) (128,2.00922e+08) (256,1.5612e+08) };
  \addlegendentry{selection=ams-multi-8,sample\_size=100000};
  \addplot coordinates { (1,2.43256e+08) (2,2.92766e+08) (4,3.19316e+08) (8,3.30146e+08) (16,3.16102e+08) (32,2.98865e+08) (64,2.72936e+08) (128,2.31843e+08) (256,1.96506e+08) };
  \addlegendentry{selection=gather,sample\_size=1000};
  \addplot coordinates { (1,2.45772e+08) (2,2.91791e+08) (4,3.15354e+08) (8,3.28989e+08) (16,2.94157e+08) (32,2.79559e+08) (64,2.35752e+08) (128,1.96891e+08) (256,1.3881e+08) };
  \addlegendentry{selection=gather,sample\_size=10000};
  \addplot coordinates { (1,2.43709e+08) (2,2.87131e+08) (4,3.03537e+08) (8,2.93067e+08) (16,2.62414e+08) (32,2.1314e+08) (64,1.5255e+08) (128,9.68156e+07) (256,2.99979e+07) };
  \addlegendentry{selection=gather,sample\_size=100000};

  \legend{} %
  \end{groupplot}
\end{tikzpicture}
\caption{Strong scaling, items per PE per second (excluding input generation).
  Legend as in \cref{fig:strong_speedup}.}
\label{fig:strong_tpp}
\end{figure*}
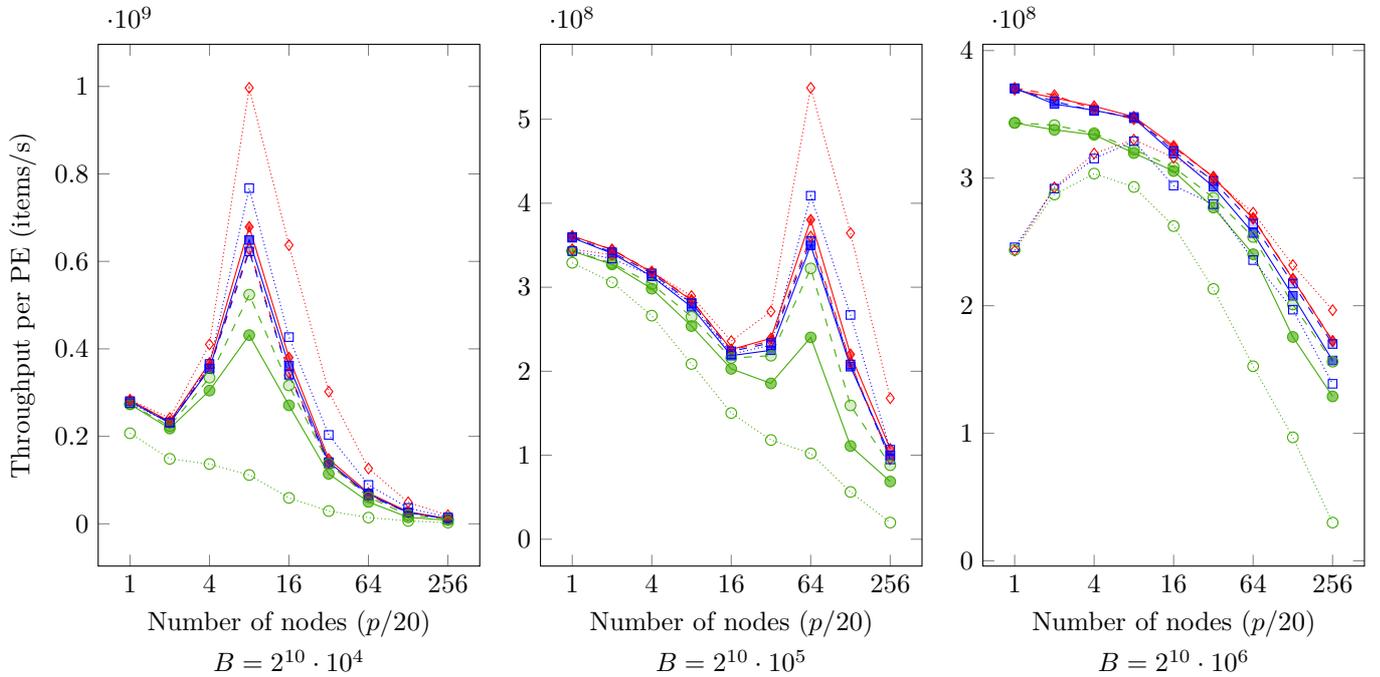

\makeatletter
\newcommand\resetstackedplots{
\makeatletter
\pgfplots@stacked@isfirstplottrue
\makeatother
\addplot [forget plot,draw=none] coordinates{(1,0) (2,0) (4,0) (8,0) (16,0) (32,0) (64,0) (128,0) (256,0)};
}
\makeatother

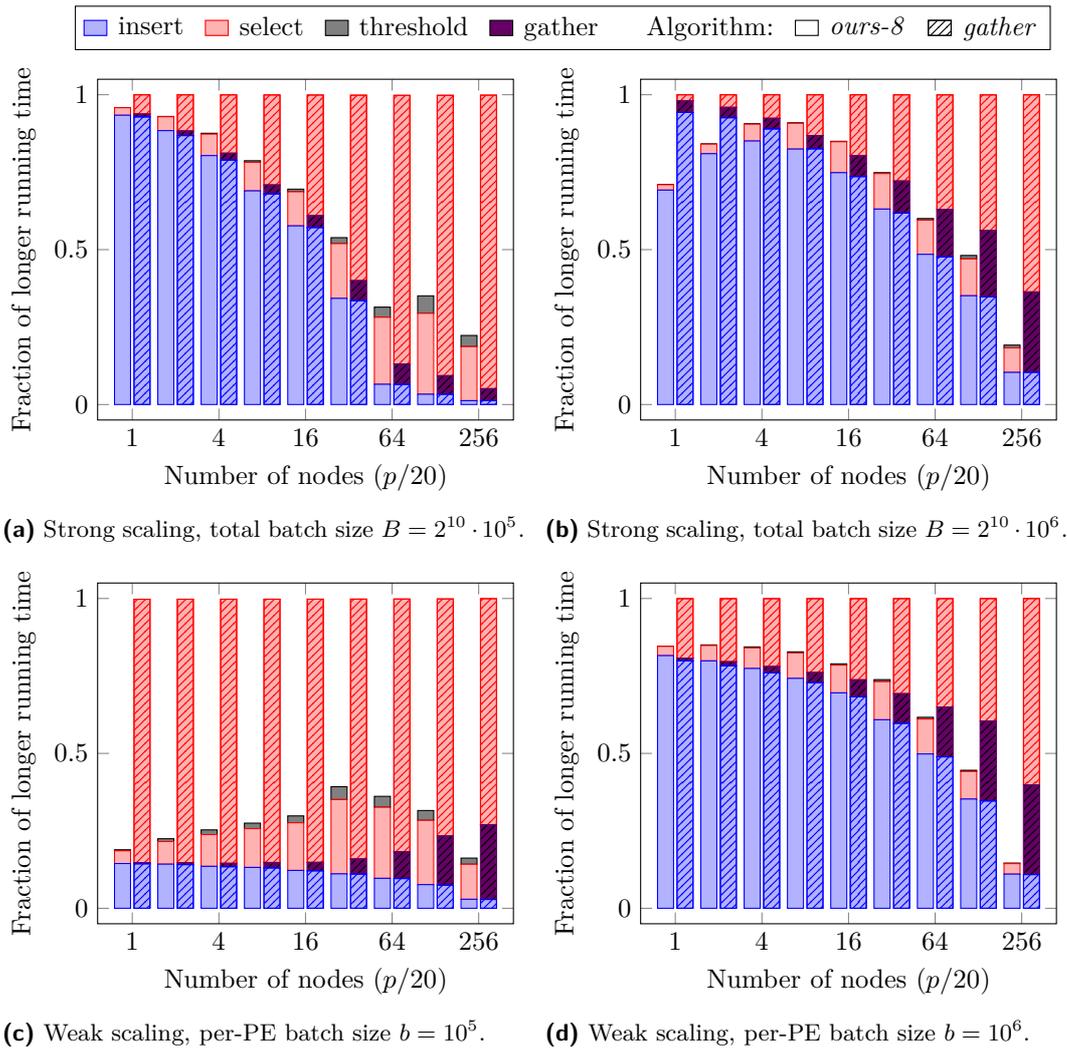
\begin{figure*}[tb!]
  \pgfplotsset{
    width=7.05cm,
    height=6.1cm,
    log basis x=2,
    log ticks with fixed point,
    ybar stacked,
    enlarge x limits=0.1,
    enlarge y limits=0.05,
    ylabel={Fraction of longer running time},
    xlabel={Number of nodes ($p/20$)},
    ymin=0,
    ymax=1,
  }
\begin{subfigure}[t]{0.49\textwidth}
\begin{tikzpicture}
  \begin{semilogxaxis}[bar width=6pt]
    \pgfplotsset{every axis plot post/.append style={mark=none}}

    \addplot[blue,fill=blue!30!white,bar shift=-3.65pt] coordinates { (1,0.933888) (2,0.884162) (4,0.803576) (8,0.689725) (16,0.576697) (32,0.34267) (64,0.0653166) (128,0.0332365) (256,0.0124459) };
    \addplot+[red,fill=red!30!white,bar shift=-3.65pt] coordinates { (1,0.0242214) (2,0.0448662) (4,0.0695019) (8,0.0919743) (16,0.110079) (32,0.176789) (64,0.216415) (128,0.26118) (256,0.174434) };
    \addplot+[black,fill=gray,bar shift=-3.65pt] plot coordinates { (1,0.000174408) (2,0.00102807) (4,0.00255182) (8,0.0052833) (16,0.00847705) (32,0.0192213) (64,0.0326773) (128,0.0563955) (256,0.0360361) };

    \resetstackedplots
    \addplot+[blue,fill=blue!30!white,bar shift=3.65pt,postaction={pattern=north east lines,pattern color=.}] coordinates { (1,0.928367) (2,0.868204) (4,0.788949) (8,0.679884) (16,0.571156) (32,0.33525) (64,0.0647421) (128,0.0329478) (256,0.0123344) };
    \addplot+[violet!80!black,bar shift=3.65pt,postaction={pattern=north east lines,pattern color=black}] coordinates { (1,0.00915461) (2,0.0146289) (4,0.0221827) (8,0.0286768) (16,0.0383763) (32,0.0643703) (64,0.0658185) (128,0.0599983) (256,0.0386747) };
    \addplot+[red,fill=red!30!white,bar shift=3.65pt,postaction={pattern=north east lines,pattern color=.}] coordinates { (1,0.0623599) (2,0.116952) (4,0.188498) (8,0.290866) (16,0.389647) (32,0.599017) (64,0.867565) (128,0.905252) (256,0.94769) };

\end{semilogxaxis}
  \begin{customlegend}[
    legend entries={insert,select,threshold,gather,Algorithm:,\emph{ours-8},\emph{gather}},
    legend style={at={(-0.3,5.2)},anchor=west, column sep=2pt,
      /tikz/column 2/.style={column sep=6pt},
      /tikz/column 4/.style={column sep=6pt},
      /tikz/column 6/.style={column sep=6pt},
      /tikz/column 8/.style={column sep=12pt},
      /tikz/column 10/.style={column sep=6pt},
      /tikz/column 12/.style={column sep=6pt}},
    legend columns=-1]
    \pgfplotsset{legend image code/.code={\draw[#1] (0cm,-0.1cm)rectangle(0.3cm,0.1cm);},}

    \addlegendimage{blue,  fill=blue!30!white}
    \addlegendimage{red,   fill=red!30!white}
    \addlegendimage{black, fill=gray}
    \addlegendimage{black, fill=violet!80!black}
    \addlegendimage{empty legend}
    \addlegendimage{black,fill=white}
    \addlegendimage{black,fill=white,postaction={pattern=north east lines,pattern color=black}}
  \end{customlegend}
\end{tikzpicture}
\caption{Strong scaling, total batch size $B=2^{10}\cdot 10^5$.}\label{fig:comps5}
\end{subfigure}
\hfill
\begin{subfigure}[t]{0.49\textwidth}
\begin{tikzpicture}
  \begin{semilogxaxis}[bar width=6pt]
    \addplot[blue,fill=blue!30!white,bar shift=-3.65pt] coordinates { (1,0.692151) (2,0.809665) (4,0.850724) (8,0.824682) (16,0.74849) (32,0.630884) (64,0.484817) (128,0.351175) (256,0.103765) };
    \addplot+[red,fill=red!30!white,bar shift=-3.65pt] coordinates { (1,0.017475) (2,0.0309119) (4,0.0543924) (8,0.0832572) (16,0.0998655) (32,0.114642) (64,0.110454) (128,0.11892) (256,0.0789868) };
    \addplot+[black,fill=gray,bar shift=-3.65pt] plot coordinates { (1,1.56595e-05) (2,0.000126242) (4,0.000353028) (8,0.000796666) (16,0.0014388) (32,0.00300818) (64,0.0051203) (128,0.0111073) (256,0.00898627) };

    \resetstackedplots
    \addplot+[blue,fill=blue!30!white,bar shift=3.65pt,postaction={pattern=north east lines,pattern color=.}] coordinates { (1,0.942762) (2,0.925816) (4,0.889513) (8,0.825272) (16,0.735732) (32,0.618084) (64,0.476551) (128,0.347109) (256,0.103831) };
    \addplot+[violet!80!black,bar shift=3.65pt,postaction={pattern=north east lines,pattern color=black}] coordinates { (1,0.0371612) (2,0.0333066) (4,0.0343326) (8,0.043204) (16,0.0671908) (32,0.103133) (64,0.153487) (128,0.214608) (256,0.25953) };
    \addplot+[red,fill=red!30!white,bar shift=3.65pt,postaction={pattern=north east lines,pattern color=.}] coordinates { (1,0.0200622) (2,0.0408544) (4,0.0761088) (8,0.131432) (16,0.196922) (32,0.278542) (64,0.369625) (128,0.437856) (256,0.636353) };

\end{semilogxaxis}
\end{tikzpicture}
\caption{Strong scaling, total batch size $B=2^{10}\cdot 10^6$.}\label{fig:comps6}
\end{subfigure}
\begin{subfigure}[b]{0.49\textwidth}
\vspace*{0.5em}
\begin{tikzpicture}
  \begin{semilogxaxis}[bar width=6pt]
    \pgfplotsset{every axis plot post/.append style={mark=none}}

    \addplot[blue,fill=blue!30!white,bar shift=-3.65pt] coordinates { (1,0.144443) (2,0.142425) (4,0.135421) (8,0.13189) (16,0.122193) (32,0.11113) (64,0.096659) (128,0.0763523) (256,0.0290075) };
    \addplot+[red,fill=red!30!white,bar shift=-3.65pt] coordinates { (1,0.0409862) (2,0.0728995) (4,0.102426) (8,0.125254) (16,0.153786) (32,0.239785) (64,0.229395) (128,0.207507) (256,0.113076) };
    \addplot+[black,fill=gray,bar shift=-3.65pt] plot coordinates { (1,0.00346089) (2,0.00913029) (4,0.0152199) (8,0.0176819) (16,0.0222183) (32,0.0416155) (64,0.0352574) (128,0.031709) (256,0.0199159) };

    \resetstackedplots
    \addplot+[blue,fill=blue!30!white,bar shift=3.65pt,postaction={pattern=north east lines,pattern color=.}] coordinates { (1,0.143517) (2,0.140889) (4,0.134082) (8,0.130158) (16,0.120888) (32,0.110333) (64,0.095905) (128,0.0745199) (256,0.0284102) };
    \addplot+[violet!80!black,bar shift=3.65pt,postaction={pattern=north east lines,pattern color=black}] coordinates { (1,0.00413197) (2,0.00611811) (4,0.010877) (8,0.0176184) (16,0.0271943) (32,0.0486541) (64,0.0864551) (128,0.158594) (256,0.241629) };
    \addplot+[red,fill=red!30!white,bar shift=3.65pt,postaction={pattern=north east lines,pattern color=.}] coordinates { (1,0.85002) (2,0.850718) (4,0.85283) (8,0.850047) (16,0.849865) (32,0.839132) (64,0.816002) (128,0.765591) (256,0.729464) };

\end{semilogxaxis}
\end{tikzpicture}
\caption{Weak scaling, per-PE batch size $b=10^5$.}\label{fig:compw5}
\end{subfigure}
\hfill
\begin{subfigure}[b]{0.49\textwidth}
\vspace*{0.5em}
\begin{tikzpicture}
  \begin{semilogxaxis}[bar width=6pt]
    \addplot[blue,fill=blue!30!white,bar shift=-3.65pt] coordinates { (1,0.815674) (2,0.798717) (4,0.774218) (8,0.742214) (16,0.695336) (32,0.60861) (64,0.49813) (128,0.35263) (256,0.110299) };
    \addplot+[red,fill=red!30!white,bar shift=-3.65pt] coordinates { (1,0.0298179) (2,0.0496756) (4,0.0665153) (8,0.0822379) (16,0.0896446) (32,0.123027) (64,0.112773) (128,0.0892376) (256,0.0338816) };
    \addplot+[black,fill=gray,bar shift=-3.65pt] plot coordinates { (1,0.000679448) (2,0.00185119) (4,0.0030025) (8,0.00337806) (16,0.00414374) (32,0.00649994) (64,0.00565053) (128,0.00356965) (256,0.00173952) };

    \resetstackedplots
    \addplot+[blue,fill=blue!30!white,bar shift=3.65pt,postaction={pattern=north east lines,pattern color=.}] coordinates { (1,0.799026) (2,0.782926) (4,0.759888) (8,0.728386) (16,0.682421) (32,0.596998) (64,0.48868) (128,0.346359) (256,0.10876) };
    \addplot+[violet!80!black,bar shift=3.65pt,postaction={pattern=north east lines,pattern color=black}] coordinates { (1,0.00862) (2,0.0127882) (4,0.0202939) (8,0.0325861) (16,0.0549417) (32,0.0955043) (64,0.160821) (128,0.257618) (256,0.289922) };
    \addplot+[red,fill=red!30!white,bar shift=3.65pt,postaction={pattern=north east lines,pattern color=.}] coordinates { (1,0.191896) (2,0.203829) (4,0.21937) (8,0.238604) (16,0.262241) (32,0.30715) (64,0.350213) (128,0.395818) (256,0.601252) };

\end{semilogxaxis}
\end{tikzpicture}
\caption{Weak scaling, per-PE batch size $b= 10^6$.}\label{fig:compw6}
\end{subfigure}
\caption{Composition of running time, normalized to the slower algorithm; strong
  scaling with $B=bp=2^{10}\cdot 10^5$ (top left) and $2^{10}\cdot 10^6$ (top
  right) items; weak scaling with $b=10^5$ (bottom left) and $10^6$
  (bottom right) items per PE; sample size $k=10^5$. Left bars: our algorithm
  using selection with 8 pivots \emph{(ours-8)}, right bars with diagonal lines:
  centralized algorithm \emph{(gather)}.}
\label{fig:composition}
\end{figure*}

\subsection{Strong Scaling}\label{exp:strong}

Our strong scaling experiments use the same sample sizes of
$k=10^3$, $10^4$, and $10^5$ items, while keeping the global batch
size $B=b\cdot p$ fixed at $B=2^{10}\cdot 10^4 \approx 10^7$,
$2^{10}\cdot 10^5 \approx 10^8$, and $2^{10}\cdot 10^6 \approx 10^9$ items.
These sizes are chosen so that the number of PEs $p$ divides them evenly for the
power-of-two numbers of 20-core machines used in our experiments.

\Cref{fig:strong_speedup} shows the relative speedups for a strong scaling
experiment with the above configurations.  We again observe that using multiple
pivots for selection (\emph{ours-8}) provides significant benefits for large sample sizes (green
lines) and does not make much difference for smaller ones (red and blue lines).
The centralized algorithm again works well only for small sample sizes and fails
to provide any significant speedup for large samples.

We also see that as long as the local batch size (\ie input size per PE) is too
large to fit into cache -- more than around $10^5$ items -- speedups increase
well, before abruptly jumping once local processing happens in the CPUs' caches.
For the smallest batch size, this effect even exceeds what would be the ideal
speedup by a significant margin.  This is a classical superlinear speedup due to
larger available cache resources.  Once the data fits into cache, local
processing time -- which previously was a significant factor -- collapses, and
now represents only a very small part of the overall time.  Initially, this
leads to a superlinear speedup.  As communication in the selection process
becomes the dominant factor in overall running time, speedups slowly decline as
the $\Ohsmash{\log^2(kp)}$ messages required for the selection dominate running time
as the number of PEs grows.  \Cref{fig:strong_tpp} shows the throughput per PE,
\ie how many items are processed at every PE per second, and confirms this.  It
clearly shows the momentary advantage of processing inputs that just fit into
cache, but are large enough to keep the fraction of running time spent on
selection low.  Once this advantage passes, the decline in throughput per PE
once again follows the predicted curve, dominated by the communication cost of
selection.

\subsection{Running Time Composition}\label{exp:composition}

\Cref{fig:composition} shows the composition of running times for our weak and
strong scaling experiments with two different batch sizes and $k=10^5$ samples.
Each pair of bars -- our algorithm using selection with 8 pivots
(\emph{ours-8}), and to its right, the centralized gathering algorithm
(\emph{gather}), marked with diagonal lines -- is normalized to the slower of
the two algorithms, which is always \emph{gather} in these experiments.

\Cref{fig:comps5,fig:comps6} present the results for strong scaling.  We can see
that the fraction of time spent on processing the local input declines as
expected, and selection becomes the dominant factor in our algorithm.  In the
centralized algorithm, however, the amount of time spent on gathering the
candidate items grows rapidly, especially for the larger batch size
(\cref{fig:comps6}).  For the smaller batch size, sequential selection dominates
the centralized algorithm's running time when using many nodes, as only
$b=B/p=20\,000$ items are processed per PE and batch when using 256 nodes (5120
PEs), which is much faster than selecting the $k=10^5$ smallest values out of
little more than $k$ candidates (the $10^5$ previously best items plus fewer
than 300 new candidates per mini-batch on average for 128 and 256 nodes in this
experiment).

The results for weak scaling are shown in \cref{fig:compw5,fig:compw6}.  For the
smaller batch size (\cref{fig:compw5}), selection clearly dominates the
gathering algorithm's running time from the beginning.  Our algorithm is consistently more than
twice as fast, and typically up to four times faster.  While this gap shrinks
for 32 and 64 nodes (640 and 1280 PEs, respectively), it grows again as the
centralized algorithm's time spent on gathering the candidates increases, even though on
average, less than $0.1$ new candidates per PE are gathered.  The
results for the larger batch size, shown in \cref{fig:compw6}, are also as
expected.  While our algorithm requires slightly more time for
sequential processing -- the candidates have to be inserted into a B+ tree
instead of stored in an array -- the centralized algorithm's selection and
gathering become unsustainably slow for large numbers of nodes.

\section{Conclusions}\label{concl}

We presented a communication-efficient reservoir sampling algorithm in the
distributed mini-batch model of streaming algorithms.  Our algorithm can handle
weighted as well as uniform inputs, and performs well both in theory and in
practice.  Experimental results show good performance, proving the usefulness of
our model and underlining that communication efficiency is of practical
importance even for node counts in the low hundreds.

A natural avenue for future work would be to study whether our approach could be
applied to sampling from a sliding window, \ie sampling only from the
(mini-batches containing the) $w$ most recently seen items. Additionally,
improvements to the selection process, such as lower recursion depth, would
directly improve the throughput of our method, especially for small mini-batch
sizes (see \cref{exp}).  Curiously, while using multiple pivots reduces the
average recursion depth in the selection by a factor of around 2.5, the running
time benefit is currently much more limited at around $35\,\%$ of selection
running time (see \cref{exp:weak}).  Preliminary measurements suggest that the
reduced number of MPI collective operations does not translate into an according
reduction in running time.  Careful engineering might be able to improve this.
Additionally, experiments on larger supercomputers could further underline that
centralized algorithms do not scale beyond a certain size.

\subsection*{Acknowledgment}\label{ack}
This work was performed on the supercomputer ForHLR funded by the Ministry of
Science, Research and the Arts Baden-Württemberg and by the Federal Ministry of
Education and Research.

The authors would like to thank Timo Bingmann, Yaroslav Akhremtsev, and Tobias
Theuer for their work on the implementation of the B+ tree.

\bibliographystyle{plainurl}
\bibliography{diss}

\end{document}